%% file: ICDCN-ByzLinSyncRegs.tex
\documentclass[sigconf]{acmart}
\AtBeginDocument{%
  }

\setcopyright{acmlicensed}
\copyrightyear{2027}
\acmYear{2027}
\acmDOI{ }
\acmConference[ICDCN 2027]{ }
\acmISBN{ }

\usepackage{wrapfig} 
\usepackage{algorithm}
\usepackage{algpseudocode}
\usepackage{xspace}
\usepackage{todonotes}
\usepackage{multicol}
\usepackage{multirow}
\usepackage{enumitem}
\theoremstyle{definition}
\newtheorem{definition}{Definition}
\usepackage{graphicx}
\usepackage{hyperref}
\input{macros}
\newtheorem{observation}{Observation}

\begin{document}

\title{Byzantine-linearizable Single-Writer Multi-Reader Register in the GPU-Inspired Shared Memory Model}

\author{Chryssis Georgiou}
\affiliation{%
  \institution{University of Cyprus}
  \country{Cyprus}}
\email{}

\author{Sathya Peri}
\affiliation{%
  \institution{IIT Hyderabad}
  \city{Hyderabad}
  \country{India}}
\email{}

\author{Manaswani Piduguralla}
\affiliation{%
  \institution{IIT Hyderabad}
  \city{Hyderabad}
  \country{India}}
\email{}

\author{Yashi Rastogi}
\affiliation{%
  \institution{IIT Hyderabad}
  \city{Hyderabad}
  \country{India}}
\email{}
\renewcommand{\shortauthors}{Georgiou et al.}

\renewcommand{\shortauthors}{A.A}

\begin{abstract}
The construction of fault-tolerant atomic registers is a fundamental problem in distributed computing. This paper studies the problem in the context of GPUs, which have become a primary platform for scientific computing and machine learning, yet increasingly face Byzantine thread behaviour arising from multi-tenant execution. A central challenge in this setting is the absence of access control. Byzantine-tolerant registers have been studied extensively in both shared-memory and message-passing systems, and most of this work assumes that only the designated writer may modify the register. GPU shared memory offers no such ownership; threads within a kernel may freely access shared objects, so a Byzantine thread may invoke write operations on any shared object, including the register itself. This raises a fundamental question: {\em if any thread can overwrite shared memory, and the writer itself may be Byzantine, can atomic register semantics still be enforced?}

We study this question in the GPU-inspired shared-memory model of Georgiou \emph{et al.}, which captures a single thread block on one streaming multiprocessor under phase-based, lock-step execution, and present a family of implementations of a Byzantine-linearizable single-writer multi-reader (SWMR) register. \AReg layers authentication over a crash-tolerant register using only constant memory, at the cost of a minimal, quantifiable probability that a read observes an invalid value. \snapqueue instead retains a bounded window of recent authenticated writes, driving this probability to zero while tolerating up to $n-1$ Byzantine processes, in both lock-free and wait-free variants. An idealized signature-aware memory, which verifies signatures within the object itself, marks the upper bound of this design space. Our results show that the GPU architectural features of bounded phases and lock-step execution can partially substitute for the absence of access control.
\end{abstract}

\ccsdesc[500]{Theory of computation~Distributed algorithms}
\ccsdesc[300]{Theory of computation~Concurrent algorithms}

\keywords{Byzantine linearizability, SWMR register, GPU-inspired shared memory, Byzantine fault tolerance, wait-freedom, authentication}
\maketitle
\vspace{-0.2cm}
\section{Introduction}
\label{sec:Intro}

\input{Introduction}
\vspace{-0.5em}

\section{Related work}
\label{sec:RelatedWork}
\input{RelatedWork}
\vspace{-0.5em}

\section{System Model}
\label{sec:model}

We follow the GPU-inspired shared-memory model introduced by Georgiou et al.~\cite{DBLP:conf/europar/GeorgiouPP25}; we do not amend the model, but restate it here to introduce notation and clarify the assumptions used in our analysis. Justification of the model's assumptions can be found in~\cite{DBLP:conf/europar/GeorgiouPP25}. The model is grounded in the behavior of real GPU hardware, capturing key execution features such as warp scheduling and lock-step execution by focusing on a single Streaming Multiprocessor (SM) executing a single thread block comprising multiple concurrent processes that share on-chip memory.

\input{SystemModel}

\section{Byzantine-Linearizable SWMR Register Using a Crash-tolerant Atomic Register with Authentication}
\label{sec:atReg}
\input{AReg}
\vspace{-0.5em}

\section{Byzantine-Linearizable SWMR Register Using a \snapqueue}
\label{sec:queue}
\input{CirclularQue}
\footnotetext{$P(\bot)$ when Crash-Tolerant Atomic Register and the scheduler are unbiased.}


\input{conclusion}

\bibliographystyle{ACM-Reference-Format}
\bibliography{references}

\newpage
\appendix
\section{Appendix}
\label{sec:app}

\subsection{Authentication In GPU Systems}
\label{sec:appAuth}
\subsubsection{Authentication using Signatures}
Digital signatures are cryptographic primitives used to authenticate digital messages and ensure their integrity. In a public-key signature scheme, each authorized writer process is provisioned with a key pair consisting of a private signing key and a corresponding public verification key. The private key is kept secret by the writer and is used to generate a signature over a message, while the public key is distributed to readers for verification.

A signature is cryptographically bound to the specific message being signed: even a single-bit modification to the message results in a different verification outcome, causing the signature check to fail. Consequently, Byzantine readers cannot alter a valid message or fabricate a new authenticated message without knowledge of the writer's private key. Under standard security assumptions, forging a valid signature without the private key is computationally infeasible. Digital signatures can be realized using well-established public-key cryptographic schemes such as RSA~\cite{rivest+:RSA:ACM} or DSA~\cite{menezes+:cryptoBook:1996}.

\subsubsection{Authentication with the Help of CUDA Kernel}
An alternative lightweight authentication mechanism leverages CUDA's built-in thread indexing variables. Within a CUDA kernel, each thread can be uniquely identified using a combination of special registers such as \texttt{blockIdx}, \texttt{blockDim}, and \texttt{threadIdx}~\cite{david+:PMPP:science:2017}. For example, a globally unique thread identifier can be computed as \texttt{tid = blockIdx.x * blockDim.x + threadIdx.x}. These built-in variables are initialized by the CUDA runtime at kernel launch and are exposed as read-only to the executing thread; they cannot be modified by user code~\cite{david+:PMPP:science:2017}.

Given that object's write method implementation itself cannot be altered by Byzantine threads, the method can append the computed unique thread identifier to the value before writing it to the register. Because threads cannot forge or alter their hardware-assigned indices within a given kernel execution, a malicious thread cannot impersonate another thread's identifier. Readers can therefore verify that the stored identifier matches the pre-designated writer thread ID and discard any value that does not match.

\subsection{SWMR Byzantine Linearizability in the GPU-inspired Model}
\begin{table*}[h]
\centering
\small
\begin{tabular}{|c|p{10cm}|}
\hline
\textbf{Symbol} & \textbf{Meaning} \\
\hline
\( n \) & Total number of processes (threads) in the system \\ \hline
\( R = \{R_1,R_2,\dots,R_n\} \) & Set of all reader processes \\ \hline
\( p \) & Number of reader processes scheduled per phase \\ \hline
\( f \) & Maximum number of Byzantine (faulty) processes \\ \hline
\( C \) & Set of correct reader processes, with \( C \subseteq R \) \\ \hline
\( F \) & Set of Byzantine (faulty) reader processes, with \( F \subseteq R \), \( |F| = f \), and \( C \cap F = \emptyset \) \\ \hline
\( \prec \) & Phase precedence relation: 
if \( op_1 \) completes in phase \( j \) and 
\( op_2 \) starts in phase \( k \) with \( j < k \),
then \( op_1 \prec op_2 \) \\ \hline
\( <_H /\rightarrow  \) & Happens-before relation in history \( H \) \\ \hline
\( \Readij{i}{j}{t}{}{x} \) & Read operation by thread \( t \), starting in phase \( i \) and completing in phase \( j \) \\ \hline
\( W^{i,j}(x) \) & Write operation of value \( x \), starting in phase \( i \) and completing in phase \( j \) \\ \hline
\( \ell(op) \) / lp & Linearization point of operation \( op \) \\ \hline
\( \bot \) & Invalid or unauthenticated value \\ \hline
\( \beta \) & Number of Byzantine readers in a Byzantine phase \\ \hline
\( \gamma \) & Scheduler bias factor \\ \hline
\end{tabular}
\caption{Summary of notation used throughout the paper.}
\label{tab:notation}
\end{table*}

\label{sec:linz}

\input{ByzLin}

\input{sufficiency_lemma}
\subsection{Correctness of Crash-Tolerant Atomic Register}
\label{sec:CCAeg}
\input{CARegCorrectness}

\subsection{Analysis of Crash-Tolerant Atomic Register}
\input{AnalysisCAReg}

\subsection{\textbf{\snapqueue Implementation}}
\label{sec:extra_snapobject}
We present four implementations of the ordered object used by the
register construction. Each implementation realises the same abstract
interface but differs in its underlying data structure and operational
assumptions. The linked-list-based and wait-free circular-array-based
atomic objects are detailed in the main paper. Here, we discuss in detail
the remaining two implementations. These alternative designs highlight
different trade-offs between simplicity and implementation feasibility,
while preserving the semantics required by the algorithm.

\subsubsection{\textbf{Linked-List Based Atomic Object}}

Algorithm~\ref{alg:LinkedList}, described in the main paper
(Section~\ref{subsec:ll-atomic-object}), realizes the ordered object as a
singly linked list: a shared pointer \texttt{head} points to the most recently
inserted node, and every node stores
$\langle \mathit{val},\mathit{seq},\mathit{sig},\mathit{next}\rangle$, where
\texttt{seq} is the writer's sequence number and \texttt{sig} the signature over
$\langle \mathit{val},\mathit{seq}\rangle$.
\textsc{Insert}$(\langle v,\mathit{seq},\mathit{sig}\rangle)$ links a new node at
the head with a compare-and-swap (CAS)\footnote{\textbf{Compare-and-Swap (CAS):} CAS is an atomic hardware primitive supported by modern architectures. Given a memory location $O$.$\textsc{CAS}(\varepsilon, \nu)$ atomically compares the current value of $O$ with the expected value $\varepsilon$; if they are equal, it updates $O$ to $\nu$ and returns true, otherwise it leaves $O$ unchanged and returns false.}
retry loop, whose successful \textsc{CAS} is its linearization point;
\textsc{ReadLatest} traverses from \texttt{head} and returns the $p{+}1$ newest
entries, which the reader (Algorithm~\ref{alg:algo2}) then verifies over
$\langle \mathit{val},\mathit{seq}\rangle$, returning the valid value with the
largest sequence number. The object treats \texttt{seq} and \texttt{sig} as
opaque data and performs no authentication itself. We now establish its
correctness.

\paragraph{\textbf{Correctness of the Linked-List Atomic Object}}
\label{subsec:ll-correctness}

We show that Algorithm~\ref{alg:LinkedList} implements a linearizable
\snapqueue object. We then use this to establish
Byzantine linearizability of the resulting SWMR register.
The first three conditions of Def.~\ref{def:byzlinz} hold when the writer is correct, while the fourth condition (Intra-Phase Reader Consistency) holds in both cases, whether the writer is correct or malicious, due to the lock-step execution model.
All four conditions are established against the general adversary: since the
signature binds each value to its sequence number, a Byzantine process may
re-insert an older writer-signed entry but cannot attach a fresh sequence number
to it, so replayed entries never displace the writer's most recent value.

\begin{lemma}[Linearization Points of Algorithm~\ref{alg:LinkedList}]
\label{lem:ll-linearization}
Every completed method call of the Linked-List--Based Atomic Object (Algorithm~\ref{alg:LinkedList}) has a unique
linearization point.
\end{lemma}
\begin{proof}
Consider the two operations.

\textbf{Write (\textsc{Insert}).}
A write operation repeatedly attempts to update the shared
pointer \texttt{head} using a \textsc{CAS} instruction inside its retry loop.
The first successful \textsc{CAS} that changes the value of
\texttt{head} from $h$ to $n$
constitutes the linearization point of the write.
Since \textsc{CAS} is atomic, this defines a unique
instantaneous linearization point between invocation and response.

\textbf{Read (\textsc{ReadLatest}).}
Consider a completed read operation and let $h$ be the value returned by its atomic read of \texttt{head}.
The sequence returned by the operation consists of the first $p+1$ entries reachable from $h$.
Since insertions occur only at the head of the list, and once a node is linked its \textit{val}, \textit{seq}, \textit{sig}, and \textit{next} fields are never modified and nodes are never removed, the sublist reachable from $h$ is immutable after it is read. Any concurrent insertion after the read of \texttt{head} creates a new node that precedes $h$ and therefore cannot affect the nodes reachable from $h$. Therefore the atomic read of \texttt{head} is the linearization point of the \textsc{ReadLatest} method.
\end{proof}

\begin{lemma}[Algorithm~\ref{alg:LinkedList} is Linearizable]
\label{thm:ll-linearizable}
All completed operations of Algorithm~\ref{alg:LinkedList} (Linked-List-Based Atomic Object) are linearizable and satisfy the sequential specification of \snapqueue.
\end{lemma}

\begin{proof}
By Lemma~\ref{lem:ll-linearization}, every completed operation has a
unique linearization point (LP) that lies between its invocation and
response.
We define a total order on completed operations by ordering them
according to their linearization points. Since LPs are atomic events
(\textsc{CAS} for writes and atomic read of \texttt{head} for reads),
this order respects real-time precedence.

It remains to show that this order satisfies the sequential
specification of the \snapqueue object.

Each successful write links a new node at the head of
the list.
Because \textsc{CAS} operations on \texttt{head} are atomic,
they induce a total order on writes.
Moreover, the order of nodes in the list corresponds exactly
to this linearization order:
if
\[
\ell(\writeij{p_1}{q_1}{u}{x})
<
\ell(\writeij{p_2}{q_2}{v}{x}),
\]
then the node containing $x_v$
appears before the node containing $x_u$.

A read operation returns the first $p+1$ nodes encountered
during traversal starting from the head . The method call position in the total serial order is therefore determined by the value of \texttt{head} it observes.
Because nodes are never removed and the list only grows at
the head, each read returns a prefix consistent with sequential specification.
Therefore, Algorithm~\ref{alg:LinkedList} is linearizable.
\end{proof}

\begin{lemma}[Termination guarantee for \textsc{Insert}]
\label{lem:insertTermination}
Every method call of \textsc{Insert}$(\langle v, \mathit{seq}, sig \rangle)$ in ~\algoref{LinkedList} terminates after appending the value $v$ in a finite number of steps.
\end{lemma}
\begin{proof}
    Consider an invocation of \textsc{Insert}$(\langle v, \mathit{seq}, sig \rangle)$ in ~\algoref{LinkedList}; the procedure first allocates and initializes a new node $n$ containing the triple $\langle v, \mathit{seq}, sig \rangle$. It then repeatedly performs the following steps:
    \begin{enumerate}
        \item  reads the current value of $head$,
        \item sets $n.next$ to the observed head, and
        \item attempts to update $head$ using a CAS operation
    \end{enumerate}

The only step that may fail is the CAS operation. A CAS fails if and only if another concurrent \textsc{Insert} operation has successfully updated $head$ after it was read. Thus, every failed CAS implies that some other process has completed its \textsc{Insert} by performing a successful CAS.

In our design, each process invokes \textsc{Insert} at most once per phase. Since at most $p+1$ processes (one writer and $p$ readers) can invoke \textsc{Insert} in a phase, there can be at most $p$ competing CAS operations against a given process in that phase. Consequently, the number of possible CAS failures for a correct writer within a phase is bounded by $p$.

Each successful CAS strictly changes the value of $head$ and corresponds to exactly one completed \textsc{Insert}. Therefore, after at most $p$ competing successful insertions, no further concurrent insertions remain in that phase, and the writer's CAS must succeed.

Since CAS is atomic, the number of competing processes is finite, and the scheduler guarantees that every correct process is scheduled, the writer eventually performs a successful CAS. Hence, every invocation of \textsc{Insert}$(v)$ completes after a finite number of steps and appends $v$ to the list.
\end{proof}

\begin{lemma}[Termination guarantee for \textsc{ReadLatest()}]
Every method call of \textsc{ReadLatest()} in ~\algoref{LinkedList}  terminates in a finite number of steps.
\end{lemma}

\begin{proof}
Consider an invocation of \textsc{ReadLatest}() in ~\algoref{LinkedList}.
The operation first performs an atomic read of the shared pointer $head$. $h$ denotes the value returned by this read.

The procedure then traverses the linked list starting from $h$. First, nodes are inserted only at the head of the list, and existing $next$ pointers are never modified or removed. Thus, once $h$ is read, the sublist reachable from $h$ is well-defined and finite at that moment. Concurrent insertions may add new nodes at the head, but they cannot affect the nodes reachable from $h$.

Second, the loop explicitly bounds the traversal to at most $p+1$ nodes. Therefore, the number of iterations of the loop is bounded by $p+1$, independent of concurrent activity.

Since each iteration consists only of local pointer reads and assignments, and no retry loops or blocking primitives are involved, the number of steps executed by \textsc{ReadLatest}() is finite.

Hence, every invocation of \textsc{ReadLatest}() terminates in a finite number of steps.
\end{proof}

\begin{theorem}
Every \textsc{ReadLatest()} in ~\algoref{LinkedList} operation that linearizes after the first completed write returns a non-empty list.
\end{theorem}

\begin{proof}
Let $W_1$ denote the first completed invocation of
\textsc{Insert}($\langle v, \mathit{seq}, sig \rangle$) by the (correct) writer in
~\algoref{LinkedList}.
By definition of the linearization point of \textsc{Insert}, the successful CAS that updates $head$ to a newly created node $n$ constitutes $\ell(W_1)$.

Immediately after $\ell(W_1)$, the shared pointer $head$ is no longer $NULL$ and points to node $n$.

In our model, Byzantine processes may invoke object methods arbitrarily, including \textsc{Insert}. However:

\begin{enumerate}
    \item They cannot modify shared memory except through the object's interface.
    \item Nodes, once linked into the list, are never removed.
    \item The pointer $head$ is only modified via successful CAS operations inside \textsc{Insert}.
\end{enumerate}

After $\ell(W_1)$, although Byzantine processes may subsequently perform additional successful \textsc{Insert} operations (thereby updating $head$ to point to newer nodes), the pointer $head$ can never revert to $NULL$, and the node inserted by $W_1$ remains reachable from the list.

Therefore, the returned sequence $S$ contains at least one element. Hence, every \textsc{ReadLatest}() operation that linearizes after the first completed write returns a non-empty list.
\end{proof}

\paragraph{\textbf{Upper Bound calculation for Memory}}
\label{sec:upperBound}
A phase cannot be smaller than one instruction per process; otherwise, the system would make no progress. If we treat one store and one fetch instruction as the basic unit, the write operation of the \snapqueue requires $4(p+1) + 6$ instructions. This is because a CAS operation may fail at most $p$ times, so a process may retry at most $p+1$ times.

The read operation, measured in terms of fetch and store instructions, requires $1 + 5(p+1)$ instructions.

Since most current GPU systems have $p = 32$, we obtain a per-process upper bound of
\[
\delta = 9 + 6p = 9 + 6 \cdot 32 = 201 \text{ phases}.
\]

Substituting into
\[
\Delta = (p+1)(\delta - 1)\frac{n}{p},
\]
we get
\[
\Delta = 33 \cdot 200 \cdot \frac{n}{32}.
\]

Simplifying,
\[
\Delta \approx 207n.
\]

Thus, the overall bound is $\Delta = 207n$.\\

\paragraph{\textbf{Byzantine Linearizability of the SWMR Register Built Using the Linked-List–Based Atomic Object}}
\label{sec:ll-swmr-byzlin}
We now establish Byzantine linearizability of the SWMR register constructed on top of the Linked-List–Based Atomic Object. 
We first prove a key lemma showing that every correct read returns the writer's most recent value, then show that the required ordering properties of Definition~\ref{def:byzlinz} are satisfied, and conclude that the resulting construction is Byzantine-linearizable.

\begin{lemma}[Latest Value]
\label{lem:ll-latest}
Let the writer be correct. Every correct read returns the value of the writer's
most recent write that linearizes before the read; once the first write has
completed, no correct read returns $\bot$.
\end{lemma}
\begin{proof}
Let $W^{*}$, with pair $\langle v^{*},\mathit{seq}^{*}\rangle$, be the writer's
most recent write preceding the read's linearization point. Since the writer
inserts in every phase, each reader inserts at most once per phase
(Section~\ref{subsec:design}), and insertions precede reads within a phase, at
most $p$ entries follow the node of $W^{*}$, so it lies in the $p+1$-node window
returned to the reader. Only pairs signed by the writer pass verification, and a
Byzantine process can re-insert such a pair but cannot alter its sequence number.
As the writer's sequence numbers strictly increase, no valid entry carries a
number exceeding $\mathit{seq}^{*}$, while the entry of $W^{*}$ is in the window
and verifies. Hence the reader's largest-$\mathit{seq}$ selection returns
$v^{*}$, and after the first write the window always contains a valid entry, so
$\bot$ is never returned.
\end{proof}
\begin{lemma}[No Future Read]
\label{thm:ll-no-future}
For every correct thread $t$ and phases $i \le j$,
\[
\Readij{i}{j}{t}{}{x} \neq \bot
\Rightarrow
\exists \writeij{p}{q}{}{x}
\text{ such that } p \le j.
\]
\end{lemma}




\begin{proof}
Since $\Readij{i}{j}{t}{}{}$ is executed by a correct thread and returns
$x \neq \bot$, Lemma~\ref{lem:ll-latest} implies that $x$ is the value of the
writer's most recent write whose linearization point precedes that of the
read; denote this write by $\writeij{p}{q}{}{x}$. Thus
\[
\ell(\writeij{p}{q}{}{x}) < \ell(\Readij{i}{j}{t}{}{x}).
\]
By Lemma~\ref{lem:lp-localization}, every operation linearizes between its
invocation and its response, so
$p \le \ell(\writetwoij{p}{q})$ and $\ell(\Readthreeij{i}{j}{t}) \le j$.
Combining the three inequalities yields $p \le j$.
\end{proof}
\begin{lemma}[No Read from Distant-Past]
\label{thm:ll-no-distant}
Let $\writeij{i}{j}{u}{x}$ and $\writeij{k}{l}{v}{x}$
be two write operations issued by the correct writer such that
\[
\ell(\writeij{i}{j}{u}{x})
<
\ell(\writeij{k}{l}{v}{x}).
\]
For any correct read
$\Readij{p}{q}{t}{a}{x}$
with
\[
\ell(\writeij{k}{l}{v}{x})
<
\ell(\Readthreeij{p}{q}{t}),
\]
it holds that $a \neq u$.
\end{lemma}



\begin{proof}
By Lemma~\ref{lem:ll-latest}, the read returns the value of the writer's most
recent write preceding its linearization point. Since
$\ell(\writeij{k}{l}{v}{x}) < \ell(\Readthreeij{p}{q}{t})$, that most recent
write is $\writeij{k}{l}{v}{x}$ or a later write of the writer, and its
sequence number is at least $\mathit{seq}_v > \mathit{seq}_u$. As the writer
writes distinct values with distinct sequence numbers, the returned value is
not $x_u$. In particular, a Byzantine process may re-insert the older entry of
$x_u$ near the head, but that entry carries $\mathit{seq}_u$ and is therefore
never selected by the reader (Algorithm~\ref{alg:algo2}).
\end{proof}

\begin{lemma}[No New-Old Inversion]
\label{thm:ll-no-new-old}
Let
\[
\Readij{i}{j}{t}{a}{x}
\quad\text{and}\quad
\Readij{k}{l}{t'}{b}{x}
\]
be two reads with $j < k$.
Then $b \ge a$.
\end{lemma}

\begin{proof}
Since $j < k$, the first read completes before the second is invoked, so
\[
\ell(\Readthreeij{i}{j}{t}) < \ell(\Readthreeij{k}{l}{t'}).
\]
By Lemma~\ref{lem:ll-latest}, each read returns the value of the writer's most
recent write preceding its own linearization point. Every write of the writer
that precedes $\ell(\Readthreeij{i}{j}{t})$ also precedes
$\ell(\Readthreeij{k}{l}{t'})$, so the most recent write seen by the second
read is the same as, or later than, the one seen by the first. As the writer
writes $x_1, x_2, \ldots$ in increasing order, $b \ge a$.
\end{proof}


\begin{lemma}[Intra-Phase Reader Consistency]
\label{thm:ll-intra-phase}
Let $j$ be any phase, and let
\[
\Readthreeij{i}{j}{1}, \Readthreeij{i}{j}{2}, \ldots, \Readthreeij{i}{j}{p}
\]
be read operations executed by correct threads that all complete in phase $j$.
Then all such reads return the same value.
Formally,
\[
\Readij{i}{j}{1}{1}{x},
\Readij{i}{j}{2}{2}{x},
\ldots,
\Readij{i}{j}{p}{p}{x}
\;\Longrightarrow\;
x_1 = x_2 = \cdots = x_p .
\]
\end{lemma}

\begin{proof}
Consider any correct thread $t$ whose read operation
$\Readij{i}{j}{t}{}{}$ completes in phase $j$.
By Lemma~\ref{lem:lp-localization}, this read admits a unique
linearization point
$\ell(\Readthreeij{i}{j}{t})$ such that
$i \le \ell(\Readthreeij{i}{j}{t}) \le j$.

In Algorithm~\ref{alg:LinkedList}, a read operation first performs an atomic read
of the shared pointer \texttt{head} (Line~16).
We define this atomic read as the read's linearization point.
Subsequent traversal of the linked list (Lines~17--20)
does not modify shared state.

By the GPU-inspired system model (\secref{model}),
execution within a phase follows the SIMT lock-step discipline:
all correct threads scheduled in phase $j$
execute the same instruction sequence synchronously.
Hence, all correct reads that complete in phase $j$
perform their atomic read of \texttt{head}
in the same phase and observe the same pointer value.

Since the linked list is inserted at the head only and nodes are never removed,
the sublist reachable from the observed head pointer
is identical for all such threads.
Moreover, the traversal, the signature verification over
$\langle v,\mathit{seq}\rangle$, and the selection of the valid value with the
largest sequence number (Algorithm~\ref{alg:algo2}) are deterministic functions
of that sublist.

Therefore, all correct read operations that complete in phase $j$
return the same value. Hence, $
x_1 = x_2 = \cdots = x_p$.
\end{proof}

\begin{theorem}
\label{thm:ll-byzlin}
Algorithm~\ref{alg:LinkedList} implements a Byzantine-linearizable SWMR register.
\end{theorem}
\begin{proof}
Algorithm~\ref{alg:LinkedList} is linearizable (Lemma~\ref{thm:ll-linearizable}),
and by Lemma~\ref{lem:invocation-independence} linearizability is preserved when
method invocations are issued by Byzantine processes. By
Lemma~\ref{lem:ll-latest}, every read executed by a correct process returns the
value of the writer's most recent write preceding it, so every read returns the
value of the most recent preceding write in the induced sequential execution,
and the history observed by the correct processes extends to a legal
linearizable history, as required by Definition~\ref{def:iditDef}.
Lemmas~\ref{thm:ll-no-future}-\ref{thm:ll-intra-phase} establish the four
ordering conditions of Definition~\ref{def:byzlinz} against the general
adversary. Hence, the construction is Byzantine-linearizable.
\end{proof}
\begin{lemma}[Correctness of \textsc{ReadLatest}]
\label{lem:readlatest-correct}
For every execution of Algorithm~\ref{alg:LinkedList}, a completed
\textsc{ReadLatest}() operation returns exactly the most recent
$p+1$ values that were successfully inserted into the linked list
prior to the linearization point of the operation, ordered from
newest to oldest. If fewer than $p+1$ values exist in the list,
then all existing values are returned.
\end{lemma}

\begin{proof}
Consider an arbitrary completed \textsc{ReadLatest}() operation
$R$. Let the linearization point of $R$ be the atomic read of
\texttt{head} at line~16 of Algorithm~\ref{alg:LinkedList}.

At the linearization point, the variable \texttt{head} points to
the node corresponding to the most recently completed successful
\textsc{Insert} operation. Every \textsc{Insert} operation creates a
new node and atomically prepends it to the front of the linked
list using a successful \texttt{CAS(head,h,n)} operation.
Consequently:

\begin{enumerate}
    \item every successful insertion adds exactly one new node to
    the front of the list, and
    \item the order of nodes reachable from \texttt{head} corresponds
    exactly to the reverse order of successful insertions.
\end{enumerate}

After reading \texttt{head}, operation $R$ traverses the linked
list by repeatedly following the \texttt{next} pointers
(lines~17--19). Since insertions only prepend new nodes and never
modify existing \texttt{next} pointers, the traversal observes a
consistent prefix of the list beginning from the node referenced
by \texttt{head} at the linearization point.

The traversal stops after collecting either:

\begin{itemize}
    \item exactly $p+1$ nodes, or
    \item all nodes currently reachable from \texttt{head}$,$ if the
    list contains fewer than $p+1$ elements.
\end{itemize}

Therefore, the sequence returned by \textsc{ReadLatest}()
contains precisely the most recent $p+1$ inserted values existing
at the linearization point of $R$, ordered from newest to oldest.

Hence the lemma follows.
\end{proof}



\subsubsection{\textbf{\hwq}}
\input{HWQFULL}
\subsubsection{\textbf{Array-Based Atomic Object}}
\input{ArrayBasedAtomicObject}

\input{WaitFreeSnapObject}

\medskip
\noindent{\textbf{Resilience of the SnapObject-Based Construction}}
\input{queueanalysis}

\subsection{\textbf{Idealized Signature-Aware Memory}}
\label{sec:sign}

\input{SignatureAware}

\end{document}

%% file: macros.tex
\newcommand{\cmnt}[1]{}
\newcommand{\ignore}[1]{}
\newcommand{\remove}[1]{}

\newtheorem{theorem}{Theorem}

\newtheorem{lemma}[theorem]{Lemma}

\newcommand{\secref}[1]{Section~\ref{sec:#1}}
\newcommand{\figref}[1]{Fig.~\ref{fig:#1}}

\newcommand{\defref}[1]{Definition~\ref{def:#1}}

\newcommand{\algoref}[1]{{Algorithm \ref{alg:#1}}}

\newcommand{\mth} {method\xspace}

\newcommand{\blble} {Byzantine-linearizable\xspace}
\newcommand{\blbty} {Byzantine-linearizability\xspace}

\newcommand{\resfac} {resilience factor\xspace}

\newcommand{\wf} {wait-free\xspace}
\newcommand{\wfdom} {wait-freedom\xspace}

\newcommand{\readP} {ReadLatest\xspace}

\newcommand{\AReg} {{CAReg}\xspace}
\newcommand{\ARegister} {{Atomic Register}\xspace}
\newcommand{\snapqueue}{SnapObject\xspace}

\newcommand{\hwq} {Queue with snap\xspace}

\newcommand{\byz} {Byzantine\xspace}

\newcommand{\Readij}[5]{R^{#1,#2}_{#3}\!\left(#5_{#4}\right)}

\newcommand{\Readthreeij}[3]{R^{#1,#2}_{#3}}

\newcommand{\writeij}[4]{W^{#1,#2}\!\left(#4_{#3}\right)}

\newcommand{\writetwoij}[2]{W^{#1,#2}}

%% file: Introduction.tex
Concurrent systems coordinate processes through shared memory, and the read/write register is the most elementary shared object on which such coordination is built~\cite{Lamport1986}. 

\noindent\textbf{GPUs as a shared-memory platform.}
Graphics Processing Units (GPUs) have evolved beyond graphics rendering into a primary platform for scientific computing, artificial intelligence, and high-performance computing~\cite{david+:PMPP:science:2017}. Unlike CPUs, which minimize latency across a few complex control flows, GPUs maximize throughput across many lightweight threads that execute in lock-step under hardware scheduling, with limited support for fine-grained access control over shared memory~\cite{david+:PMPP:science:2017}. These constraints fundamentally shape how shared-memory abstractions can be realized on GPU hardware.

Designing fault-tolerant concurrent algorithms depends on the ability of shared objects to preserve their semantics under a prescribed failure model. Classical shared-memory research has primarily focused on crash failures, where faulty processes halt or are arbitrarily delayed. Modern GPU deployments, however, increasingly face arbitrary Byzantine ~\cite{Byzantine_Generals_paper} deviations, increasingly face arbitrary Byzantine ~\cite{Byzantine_Generals_paper} deviations. 
GPUs are now routinely shared: multi-tenancy mechanisms colocate kernels from different clients on the same hardware~\cite{Weaver+:SC24:2024}, and third-party device code plugins, custom ML operators, JIT-compiled kernels execute in the same address space as trusted computation. To a shared object, a buggy, compromised, or misbehaving co-resident thread is indistinguishable from a Byzantine process, so tolerating such behaviour at the level of the object itself is of direct practical interest.

To reason about correctness under such failures, Byzantine linearizability~\cite{CKDISC2021} extends classical linearizability by requiring that the behavior observed by correct processes remains consistent with that of an atomic object despite adversarial interference. Byzantine-tolerant shared objects have been extensively studied in CPU-based systems~\cite{Kshemkalyani+:Improve-Hu-Toueg:SSS:2025, HuToueg:SWMR-LieDeny:PODC:2025}. These constructions typically rely on ownership assumptions over shared objects. For example, in standard Single-Writer Multi-Reader (SWMR) register constructions~\cite{HuToueg:SWMR-LieDeny:PODC:2025}, only the designated writer is allowed to modify the register, while readers may invoke only read operations. Such ownership assumptions are incompatible with GPU architectures, where any thread may invoke the operations exposed by those objects~\cite{david+:PMPP:science:2017}.

This mismatch between classical assumptions and GPU execution semantics introduces new challenges. In the absence of access control, a Byzantine thread may arbitrarily invoke write operations on shared objects, potentially violating the semantics of standard shared objects such as SWMR registers. At the same time, the structured execution of GPUs, particularly lock-step execution and bounded scheduling, provides a form of synchrony that may help mitigate adversarial interference. This interplay between unrestricted memory access and structured execution raises a fundamental question: \emph{can atomic shared-memory objects such as SWMR registers be implemented in GPU-based execution models in the presence of Byzantine behavior?}

We study this question within the GPU-inspired shared-memory model introduced by Georgiou et al.~\cite{DBLP:conf/europar/GeorgiouPP25}, which models a single thread block on one Streaming Multiprocessor (SM).

\noindent\textbf{Contributions.} In essence, we consider different underlying base objects, each assumed to be non-corruptible, and show how, together with digital signatures, a Byzantine-linearizable SWMR register can be built on top of them in the hardware-inspired GPU model. In all but the idealized construction of Section~\ref{sec:ideal}, the base objects are oblivious to signatures: they store and return values without inspecting them, and authentication is performed entirely by the processes. Specifically, this paper makes the following contributions.
1. We formulate the problem of implementing Byzantine linearizable SWMR registers in the GPU-inspired shared-memory model (Sections~\ref{sec:model} and~\ref{Byz-Reg}).\\
 2.  We present three Byzantine-linearizable SWMR register constructions, all relying on authentication (Sections~\ref{sec:atReg}--\ref{sec:ideal}):\\
 a. \AReg\ layers authentication over a crash-tolerant atomic register and needs only constant memory, at the cost of a probabilistic resilience guarantee, which we analyze under scheduler bias and Byzantine interference (Section~\ref{sec:atReg}).\\
  b. \snapqueue\ instead retains a bounded window of recent authenticated writes, driving that probability to exactly zero (Section~\ref{sec:queue}). We realize its underlying ordered object in two ways: a \emph{lock-free} implementation based on a linked list, and three \emph{wait-free} implementations based on a circular array, a Herlihy Wing inspired queue with snapshot support, and a fixed-size atomic array. These wait-free realisations additionally guarantee that no correct process, including the writer, can be starved by an adversarial scheduler. These realizations expose different trade-offs in memory and read complexity.\\
  c. An idealized signature-aware memory abstraction characterizes the upper boundary of the design space, under stronger assumptions (Section~\ref{sec:ideal}).\\
 3.  We establish the correctness and Byzantine linearizability of all constructions, and characterize the trade-offs between resilience, memory, read complexity, and multiphase execution in the GPU-inspired model (Appendix~\ref{sec:app}).

%% file: RelatedWork.tex
Registers form the basic abstraction in shared-memory distributed systems. Depending on the guarantees provided under concurrent access, registers are classified as \emph{safe}, \emph{regular}, and \emph{atomic}~\cite{Lamport1986}. Atomic (linearizable) registers provide the strongest semantics: every operation appears to take effect instantaneously between its invocation and response while respecting real-time precedence. Atomic registers serve as the fundamental building block for higher-level concurrent objects such as snapshots, reliable broadcast, queues, consensus objects, and asset transfer systems~\cite{Herlihy+:AMP:Book:2020}.

In the classical crash-failure model, the construction of atomic registers from weaker primitives has been extensively studied~\cite{Herlihy+:AMP:Book:2020}, showing that strong consistency is achievable in fully asynchronous systems despite process crashes. Extending such guarantees to Byzantine failures is substantially harder, since faulty processes may behave arbitrarily rather than simply stop. In asynchronous message-passing systems, Kowalski et al.~\cite{SinchanAchourMostfaouibyzprone} studied Byzantine SWMR register abstractions and established tight resilience bounds for constructing stronger variants from weaker ones.

Cohen and Keidar introduced \emph{Byzantine linearizability}~\cite{CKDISC2021}, a correctness condition requiring that the projection of an execution onto correct processes can be extended into a linearizable history by augmenting it with additional Byzantine operations. They used this framework to study reliable broadcast, atomic snapshot, and asset transfer from SWMR registers.

Hu and Toueg~\cite{HuToueg:SWMR-SWSR:DISC:2022} pushed this further and showed an impossibility result central to our own motivation: there is no \emph{wait-free} linearizable implementation of a SWMR register from crash-tolerant SWMR registers under Byzantine failures, even with restricted adversarial assumptions. To circumvent this, they introduced strengthened abstractions verifiable and sticky registers that emulate signature-like guarantees directly in shared memory~\cite{HuToueg:SWMR-LieDeny:PODC:2025}, later improved by Kshemkalyani et al.~\cite{DBLP:journals/corr/abs-2405-19457}. All of the above results target fully asynchronous shared-memory or message-passing systems. GPU architectures are different: they follow a phase-based execution model with structured scheduling and lock-step execution within warps. Georgiou et al.~\cite{DBLP:conf/europar/GeorgiouPP25} introduced a GPU-inspired shared-memory model capturing these properties, but the impact of this structure on Byzantine linearizability, and on the progress guarantees a register construction can offer, has not been studied.

In this work, we investigate the construction of a Byzantine-linearizable SWMR register in the GPU-inspired shared-memory model. Unlike prior work~\cite{DBLP:journals/corr/abs-2405-19457, HuToueg:SWMR-SWSR:DISC:2022, HuToueg:SWMR-LieDeny:PODC:2025} that assumes memory access control, purely asynchronous execution, or sticky/verifiable memory, we assume none of these: any thread may overwrite shared memory, and the phase-based lock-step structure and mandatory writer scheduling of the GPU model take the place of both synchrony and sticky objects. We rely on authentication mechanisms such as digital signatures to identify legitimate writes and prevent impersonation~\cite{HuToueg:SWMR-LieDeny:PODC:2025, DBLP:journals/corr/abs-2405-19457}. We further show that this structure suffices to achieve \wfdom~\cite{Herlihy:TPLS:1991}: our \wf\ \snapqueue\ realizations order insertions with an unconditional fetch-and-increment, in the spirit of the Herlihy--Wing queue~\cite{herlihy1990linearizability}, rather than through the CAS retry loops or helping mechanisms employed by classical wait-free queue constructions~\cite{Kogan2011WaitfreeQW}. Consequently, every operation of a correct process completes in a bounded number of its own steps, despite Byzantine interference and adversarial scheduling.
\vspace{-0.5em}

%% file: SystemModel.tex
\noindent\textbf{GPU-Constrained Shared Memory Model.} Fig.~\ref{fig:SM} gives a high-level view of the system model, which is a faithful abstraction of how computation is organized on modern GPU hardware. A GPU is built from \emph{Streaming Multiprocessors} (SMs), each hosting many simple cores, on-chip shared memory, and a hardware \emph{warp scheduler} that dispatches groups of threads called \emph{warps} under the Single-Instruction Multiple-Threads (SIMT) model~\cite{david+:PMPP:science:2017,Ogier:book:2013,NVIDIA:2024}: within a warp, threads share a single program counter and execute in lock-step. Warps are grouped into a \emph{block}, bound to a single SM, where all threads share fast on-chip memory and the same warp scheduler. {\em Following~\cite{DBLP:conf/europar/GeorgiouPP25}, we focus on the execution of a single block on a single SM; in Fig.~\ref{fig:SM}, the oval highlights this subsystem.} 

\noindent\textbf{Processes.}
The system consists of $n$ processes $p_1, p_2, \ldots, p_n$, each a GPU thread within this block. Let $p$ denote the warp size (typically $p = 32$ on modern NVIDIA GPUs~\cite{NVIDIA:2024}); the block thus contains $\lceil n/p \rceil$ warps, and the hardware warp scheduler dispatches them in short, bounded execution windows.

\noindent\textbf{Failure Model.}
We consider both crash and Byzantine failures. Under the crash model, a faulty process halts and takes no further steps. Under the Byzantine model, a faulty process may deviate arbitrarily from the prescribed protocol (for instance a Byzantine reader may perform arbitrary writes). Following~\cite{DBLP:conf/europar/GeorgiouPP25}, we assume an adaptive adversary that may corrupt up to $f < n$ processes. Byzantine processes may behave arbitrarily, but the underlying hardware constrains them. In particular, they cannot bypass the warp scheduler, impersonate other processes, or directly corrupt shared memory. While a Byzantine thread may deviate from the assigned computation and execute arbitrary instructions, it can do so only when it is scheduled and only within its assigned warp. It cannot execute outside its scheduled phase or obtain additional execution beyond what the scheduler allows. \emph{These restrictions reflect the GPU execution environment: scheduling is enforced in hardware, and threads cannot fabricate execution slots or interfere with other threads' identities} (see~\cite{DBLP:conf/europar/GeorgiouPP25} for further discussion). In particular, a thread's identity derives from read-only special registers (\texttt{threadIdx}, \texttt{blockIdx}) initialised by the runtime at kernel launch and not writable by user code~\cite{david+:PMPP:science:2017}; see also Appendix~\ref{sec:appAuth}.

\noindent\textbf{Phase-Based Execution.}
Computation proceeds in an infinite sequence of \emph{phases}. In each phase, the warp scheduler selects exactly one warp and every process in that warp executes a bounded number of instructions. This bound mirrors real GPU schedulers, which issue only a small constant number of instructions per warp before switching to hide memory latency~\cite{NVIDIA:2024,Ogier:book:2013}. This also ensures that no process, including a Byzantine one, makes unbounded progress within a phase. In this paper, we use the terms \emph{phase} and \emph{warp execution} interchangeably.
\begin{figure}[h!]
    \centering
    \includegraphics[scale=0.23]{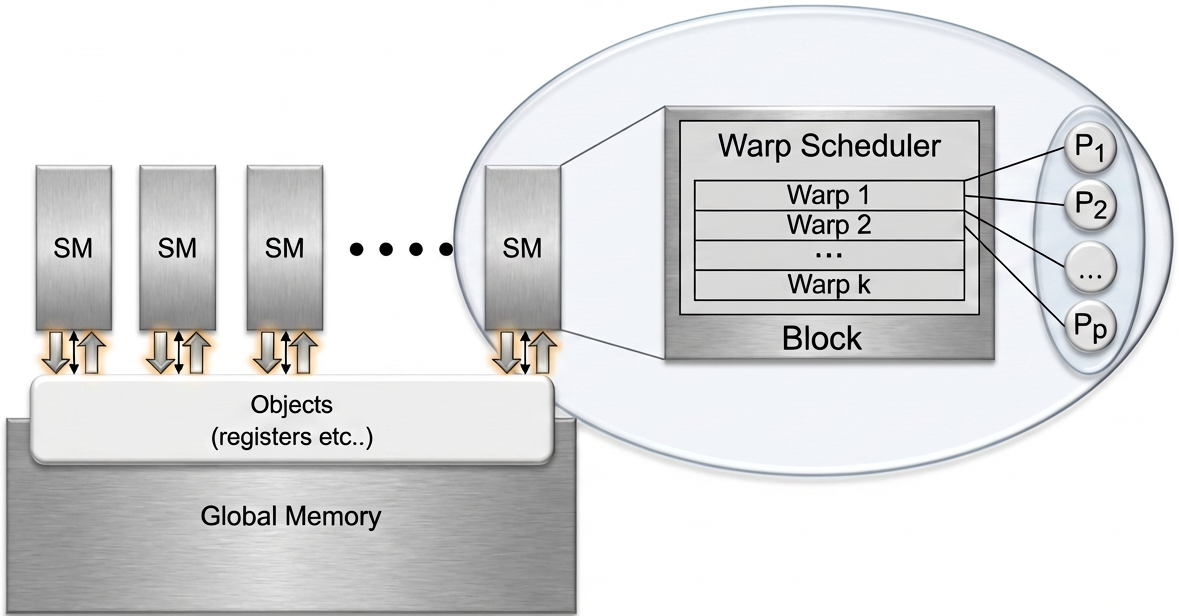}
    \caption{High-level depiction of the GPU-inspired system model}\vspace{-1em}
     \label{fig:SM}
\end{figure}
\\
\noindent\textbf{Lock-Step Execution.}
When threads within a warp diverge onto different operations, the hardware serializes execution: processes scheduled in a phase are partitioned into groups by the operation they perform (e.g., read or write), threads within a group execute in lock-step, and different groups execute sequentially within the same phase (see \figref{phase-lockstep}). Although Byzantine processes may deviate in values or operations, they cannot alter the control flow of other threads, modify program counters, or extend phase duration.

\noindent\textbf{Cycles.}
A \emph{cycle} is a sequence of phases in which every process is scheduled at least once. Since each phase schedules $p$ processes, a cycle spans at least $\lceil n/p \rceil$ phases. Cycles are used only for reasoning about fairness and progress.

\noindent\textbf{Scheduler Fairness and Bias.}
Our analysis considers both unbiased and biased warp schedulers. An \emph{unbiased scheduler} selects warps without preference, whereas a \emph{biased scheduler} may favour warps containing Byzantine processes due to effects such as latency hiding, memory-access behaviour, or co-tenant interference. We capture this behaviour using a bias factor $\gamma \geq 1$, which bounds the relative scheduling frequency of Byzantine-containing phases compared to purely correct phases. The case $\gamma=1$ corresponds to an unbiased scheduler.

\noindent\textbf{Shared Memory.}
The system provides reliable shared memory that is accessed exclusively through well-defined object interfaces~\cite{CKDISC2021, DBLP:conf/europar/GeorgiouPP25}. This restriction is standard: if Byzantine processes could arbitrarily corrupt raw memory, no meaningful correctness guarantees would be possible.

We do not assume any access control. Any process may invoke any operation exposed by an object interface, but no invocation can violate the object's semantics. This reflects GPU hardware, where threads within a block freely access shared memory without ownership restrictions~\cite{david+:PMPP:science:2017}, but must interact through well-defined instructions such as load, store, and other atomic primitives. A Byzantine process may thus invoke operations on any shared object, but it cannot reach memory outside its addressable scope, in particular, another thread's private registers and local memory, nor bypass the object interface to shared state.

\noindent\textit{Remark.}
If phase-based scheduling and lock-step execution are removed, the model reduces to a classical asynchronous shared-memory system. These constraints are precisely what distinguish the GPU model~\cite{DBLP:conf/europar/GeorgiouPP25} from the traditional asynchronous shared-memory system.
Our model captures the classical SIMT lock-step discipline of~\cite{DBLP:conf/europar/GeorgiouPP25}. On recent hardware, Independent Thread Scheduling (Volta and later) gives each thread its own program counter and call stack, weakening strict intra-warp lock-step; this primarily affects Observation~1 (intra-phase reader consistency), while the hardware-enforced scheduling and identity restrictions continue to hold. Extending our constructions to independent thread scheduling is left as future work.\\
\noindent\textbf{On Replay.}
Digital signatures authenticate the \emph{origin} of a value, not its \emph{freshness}: a valid signature proves that the designated writer produced the value at some point in the past, but not that it is the writer's most recent one. A Byzantine process may therefore copy a validly signed pair from the register and re-install it later; we call this behaviour \emph{replay}. Since the signature of a replayed pair still verifies, the re-installed value is indistinguishable from a fresh write at the level of a single observation. Accordingly, we distinguish two adversaries. A \emph{general adversary} may invoke arbitrary operations with arbitrary arguments, including the replay of previously signed values. A \emph{non-replaying adversary} is weaker: it never re-installs a validly signed pair, so every value it writes fails signature verification. Unless stated otherwise, Byzantine linearizability and our resilience bounds hold against the general adversary; the sole exceptions are the freshness and resilience bounds of the CAReg construction (Section~\ref{sec:atReg}), which assume the non-replaying adversary.
\vspace{-0.5em}
\section{Byzantine Linearizability and SWMR Register}
\label{Byz-Reg}
This section introduces the correctness conditions, shared-memory abstractions, and hardware primitives used throughout the paper.\\
\noindent\textbf{Histories and Operations.}
Executions are modeled using \emph{histories}. An \emph{operation} consists of an invocation event and a matching response event; an invocation without a matching response is called \emph{pending}. A \emph{history} $H$ is a finite sequence of invocation and response events and is \emph{well-formed} if, for every process, invocation and response events alternate, beginning with an invocation. Let $\mathit{complete}(H)$ denote the history obtained from $H$ by removing all pending invocations. Two histories are said to be \emph{equivalent} if they induce identical per-process subhistories. A history is \emph{sequential} if no two operations overlap; that is, for any two operations $o_1$ and $o_2$, either the response of $o_1$ occurs before the invocation of $o_2$, or the response of $o_2$ occurs before the invocation of $o_1$. A sequential history is \emph{legal} if it satisfies the sequential specification of the object.

A history $H$ also induces a \emph{real-time precedence order}, denoted by $<_H$, on operations. For two operations $o_1$ and $o_2$, we write $o_1 <_H o_2$ iff the response event of $o_1$ occurs before the invocation event of $o_2$ in $H$.\\
\noindent\textbf{Linearizability.}
Linearizability~\cite{herlihy1990linearizability} is a correctness condition for concurrent objects requiring that every operation appears to take effect atomically at some instant between its invocation and response events.
\vspace{-0.5em}
\begin{definition}[Linearizability~\cite{herlihy1990linearizability}]
A history $H$ is \emph{linearizable} if it can be extended to a history $H'$ (by appending response events to some pending invocations) such that:
\begin{itemize}
    \item[\textbf{L1:}] $\mathit{complete}(H')$ is equivalent to a legal sequential history $S$, and
    \item[\textbf{L2:}] $<_H \subseteq <_S$.
\end{itemize}
The sequential history $S$ is called a \emph{linearization} of $H$.
\end{definition}

\noindent\textbf{Byzantine Linearizability.}
Byzantine linearizability extends linearizability to systems in which processes may exhibit Byzantine failures. Although faulty processes may behave arbitrarily, the behavior observed by correct processes must remain consistent with that of an atomic object. Cohen and Keidar~\cite{CKDISC2021} introduced Byzantine linearizability as a correctness condition for shared objects in asynchronous systems with Byzantine process failures.

\begin{definition}[Byzantine Linearizability (Cohen--Keidar)~\cite{CKDISC2021}]
A history $H$ is \emph{Byzantine linearizable} if there exists a history $H'$ such that:
$H'|_{\text{correct}} = H|_{\text{correct}}$
and $H'$ is linearizable.
\label{def:iditDef}
\end{definition}
Here, $H$ denotes the complete execution history consisting of invocation and response events of all processes, while $H|_{\text{correct}}$ denotes the projection of $H$ onto correct processes obtained by removing all events generated by Byzantine processes.


\noindent\textbf{Register.}
A register stores values from a domain $V$ and supports
\textsc{write}$(v)$ and \textsc{read}() operations. A
\textsc{write}$(v)$ operation updates the register state to
$v \in V$, while a \textsc{read}() operation returns the value
written by the most recent preceding \textsc{write}, or the
initial value $v_0$ if no write has occurred in case of sequential execution.\\
\noindent\textbf{Atomic Registers.}
An atomic register is a linearizable read/write object that supports concurrent read and write operations. A \textsc{read} operation returns the value written by the most recent preceding \textsc{write} operation~\cite{Lamport1986}. In the context of registers we use the terms atomic and linearizable interchangeably.\\
\noindent\textbf{SWMR Registers.}
A \emph{Single-Writer Multiple-Readers (SWMR) register} supports one writer and multiple readers through \textsc{write} and \textsc{read} operations. The sequential specification requires every \textsc{read} to return the value written by the most recent preceding \textsc{write}. \\
\noindent\textbf{Non-Corruptible Crash-Tolerant Atomic Register.}
A \emph{crash-tolerant atomic register} is a linearizable implementation of a register that guarantees that every operation invoked by a correct process eventually completes, even if some processes crash or halt prematurely during execution. Throughout the paper, such registers, and all other base objects we use, are assumed to be \emph{non-corruptible}: the underlying memory cannot be corrupted, and the only way any process, correct or Byzantine, can access it is by invoking the operations of these base objects. A Byzantine process may thus invoke arbitrary operations on a base object, but cannot alter its state other than through those operations. Except for the idealized object of Section~\ref{sec:ideal}, base objects are oblivious to signatures: they store and return values without inspecting or verifying them.\\
\noindent\textbf{Byzantine-Tolerant Linearizable Register.}
A \emph{Byzantine-tolerant linearizable register} is a read/write
object whose operations satisfy Byzantine linearizability;
that is, despite Byzantine process failures, the behavior
observed by all correct processes is consistent with a
linearizable execution of a sequentially specified register.\\
\noindent\textbf{Hardware Instructions.}
This work uses two widely supported atomic read-modify-write primitives:

\begin{itemize}
    \item \textbf{Compare-and-Swap (CAS):} Given a memory location $m$, the operation $m.\textsc{CAS}(\varepsilon,\nu)$ atomically compares the current value of $m$ with the expected value $\varepsilon$. If they are equal, the value of $m$ is updated to $\nu$ and the operation returns \texttt{true}; otherwise, $m$ remains unchanged and the operation returns \texttt{false}.

    \item \textbf{Fetch-and-Increment (FAI):} The operation
$m.\textsc{FAI}()$ atomically increments the current value
of $m$ by one and returns its previous value.
\vspace{-1.0em}
\end{itemize}
\label{sec:lockstepDiagram}
\begin{figure}[h!]
    \centering
    \includegraphics[width=\linewidth]{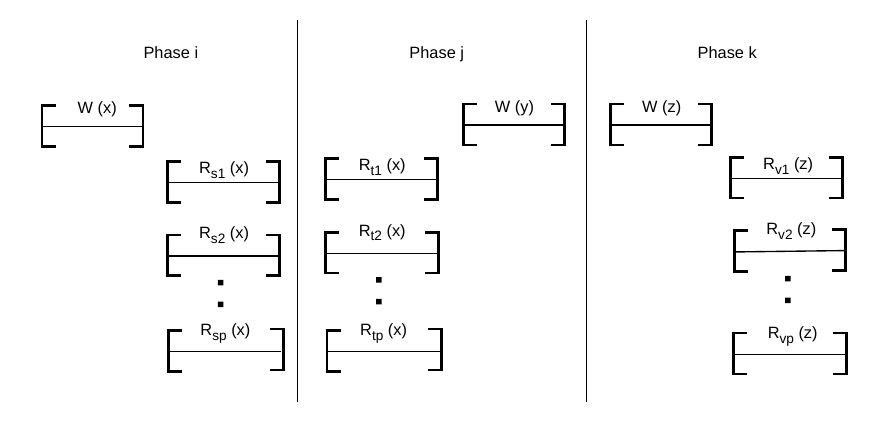}
    \vspace{-2.5em}
    \caption{Example of phase-based lock-step execution of write and read operations, where $R_t(x)$ denotes a read operation by thread $t$ returning value $x$, and $W(x)$ denotes a write operation of value $x$.}
     \label{fig:phase-lockstep}
\end{figure}
\begin{observation}
[Intra-Phase Reader Consistency in the
Lock-Step Setting]
\label{obs:intraphase}
Due to the SIMT lock-step execution semantics of the
GPU-inspired model, all correct reader processes executing
the same read instruction within a phase observe the same
shared-memory state. Consequently, as illustrated in
Figure~\ref{fig:phase-lockstep}, all correct readers
completing their read operation within the same phase return
the same value.
\end{observation}

%% file: AReg.tex
We present a simple construction of a Byzantine-linearizable SWMR register in the GPU-inspired model, built on a crash-tolerant atomic register together with authentication to filter adversarial writes.\\
\noindent\textbf{Why Crash-Tolerant Objects are Useful in the Byzantine Setting?}
Building a Byzantine-linearizable object on top of a crash-tolerant atomic object may initially appear counterintuitive, since crash-tolerant objects are not designed to tolerate arbitrary behavior. In classical asynchronous shared-memory systems, Byzantine processes may indefinitely interfere with correct processes, leading to strong impossibility results for Byzantine-linearizable registers \cite{HuToueg:SWMR-SWSR:DISC:2022}.
(Crash-tolerant atomic registers are defined in Section~\ref{Byz-Reg}, with their formal properties in Appendix~\ref{sec:CCAeg}.)
As already mentioned in Section~\ref{sec:model}, in the GPU-inspired model, Byzantine processes can interact with the underlying atomic register only through its exported operations. They cannot corrupt the internal object state, impersonate other threads, or execute outside scheduler-controlled phases. Moreover, the execution proceeds in bounded lock-step phases, and the writer participates in every phase. These structural constraints limit adversarial interference and allow crash-tolerant atomic objects to serve as a foundation for Byzantine-linearizable constructions.

\subsection{Design Features}
\label{subsec:CARegAssumption}
Our construction is derived from GPU-style execution semantics, where threads execute in lock-step under the SIMT model and warps are scheduled by hardware. The assumptions used in this work therefore capture high-level GPU execution behaviour rather than exact hardware-specific scheduling rules.

Importantly, the construction does not require precise knowledge of the scheduler implementation. We analyse the construction under both biased and unbiased schedulers and registers.

\noindent\textbf{(1) Participation per Phase.}
The model does not provide access control: any process may invoke write operations on shared memory. Hence, Byzantine readers may overwrite the register contents with arbitrary values (through the register interface). Generally, sticky or verifiable objects are often required to prevent Byzantine processes from permanently erasing or overwriting legitimate values~\cite{attie:IPL:2002}. To prevent legitimate values from being permanently erased in the absence of sticky objects, we require that the writer participate in every phase. In each phase, the writer rewrites its current value. The writer’s repeated reinforcement ensures that every phase contains at least one legitimate write attempt if the writer is honest. Each phase consists of: \textbf{ one writer}, and \textbf{ exactly $p$ readers}. The parameter $p$ refers only to readers; the writer is not included in this bound.

\noindent \textbf{(2) Shared Memory.}
The system provides a reliable, non-corruptible shared-memory abstraction accessible to all processes. The memory itself cannot be corrupted and is accessed only through well-defined object interfaces. 

\noindent \textbf{(3) Crash-tolerant \ARegister (\AReg).} Communication between processes is realized through a single non-corruptible crash-tolerant atomic register supporting read and write operations. The register is oblivious to signatures: it stores whatever value is written to it, and signature verification is performed by the readers.
The register is assumed to be linearizable and \emph{wait-free}\footnote{Wait-free: A concurrent object implementation is wait-free if each method call finishes
in a finite number of steps. Later description is in the appendix \ref{sec:CCAeg}}  \cite{Herlihy:TPLS:1991} under crash failures. Byzantine processes are restricted to invoking the register’s methods; they cannot alter its internal state except through well-defined method calls, nor can they violate its atomicity guarantees.
Consequently, although Byzantine processes may invoke operations arbitrarily (including issuing spurious writes), they cannot break the linearizability or wait-freedom of the underlying object. This structural restriction allows crash-tolerant atomic objects to serve as safe building blocks for constructing Byzantine-tolerant abstractions in this model (see~\cite{DBLP:conf/europar/GeorgiouPP25} for a detailed explanation).

\noindent\textbf{(4) Authentication.}
Since access control is absent, Byzantine readers may attempt arbitrary writes. To distinguish writer-generated values from Byzantine ones, we use signatures in this work. A correct writer attaches a cryptographic signature to each write, and correct readers (who know the signature of the writer thread in advance) accept only those values whose signatures verify successfully. A \byz writer on the other hand may not follow these steps. We discuss the implementation of the authentication mechanisms in GPUs in detail in Appendix \ref{sec:appAuth}.\\
\subsection{Algorithmic Solution}
\label{subsec:areg-soln}
The register is implemented using a single shared variable $M$ is described in ~\algoref{SWMR_AReg}, with initial value set to $\langle NULL \rangle$. The $write(v)$ method first computes a digital signature on $v$ using its private key and writes the pair $\langle v, sig \rangle$ to the atomic register $M$. Each write therefore replaces the previous contents of the register with a signed value. A reader executing $\mathrm{Read}()$ simply reads the current contents of $M$. If the stored value contains a valid signature with the writer’s public key, the reader returns the corresponding value $v$. Otherwise, it returns $\bot$. In particular, values written without a valid signature, including arbitrary bit-strings or unsigned values \texttt{NULL}, fail verification and are read as $\bot$; returning $\bot$ is part of the register's sequential specification over $V \cup {\bot}$, not a violation of it.
\begin{algorithm}[h!]
\caption{CAReg: Byzantine-Tolerant SWMR Register}
\label{alg:SWMR_AReg}
\small

\textbf{Shared variable:} $M \gets \langle\text{NULL}\rangle$


\begin{minipage}[t]{0.34\columnwidth}
\begin{algorithmic}[1]
\Procedure{Write}{$v$}
    \State $x \gets \langle v,sig\rangle$
    \State $M \gets x$
\EndProcedure
\end{algorithmic}
\end{minipage}
\hfill
\begin{minipage}[t]{0.62\columnwidth}
\begin{algorithmic}[1]
\Procedure{Read}{}
    \State $x \gets M$
    \If{$x=\langle v,sig\rangle$ for some $v$}
        \State \Return $v$
    \Else
        \State \Return $\bot$
    \EndIf
\EndProcedure
\end{algorithmic}
\end{minipage}

\end{algorithm}

Since the system executes in phases, exactly one writer and $p$ readers access the register concurrently in each phase. All correct readers scheduled in the same phase execute in lock-step and perform identical instructions. As a result, readers in the same phase observe the same register state and return the same value. If Byzantine readers overwrite the register with invalid values, correct readers return $\bot$. However, since the writer participates in every phase and rewrites its value, a valid signed value is reintroduced in each phase. This periodic rewriting ensures that correct readers continue to have opportunities to observe legitimate values despite adversarial interference.

The construction thus combines three elements: (a) atomicity provided by $\mathrm{CAReg}$, (b) authenticity enforced through digital signatures, and (c) structured phase-based execution.

Together, these mechanisms allow the register to preserve atomic semantics for correct processes in the presence of Byzantine behaviour.

\subsection{Correctness Overview}
Atomicity follows from the linearizability of the underlying crash-tolerant
register, which is preserved under Byzantine invocations; authentication
ensures \emph{value validity}: a correct reader never returns a value the writer did not produce, and lock-step execution yields intra-phase reader
consistency. Value validity alone suffices for Byzantine linearizability
against the general adversary, the construction additionally satisfies the ordering conditions of
Definition~\ref{def:byzlinz}. The full proofs appear in
Appendix~\ref{sec:CCAeg}.

\subsection{Probabilistic Resilience Analysis}
\label{subsec:CAReg_analysis}
Although the construction preserves Byzantine-linearizable semantics, Byzantine processes may repeatedly overwrite the shared register with invalid values before correct readers observe the writer's authenticated update. We therefore analyze its probabilistic resilience under varying scheduler and register behaviours.

We define the \emph{\resfac}, denoted $P(\bot)$, as the probability that a correct reader returns $\bot$ during a cycle (a sequence of phases in which every process is scheduled at least once).\\
\noindent\textbf{System Parameters.}
The resilience expressions depend on $n$, $f$, $p$, $\gamma$ and $\beta$, where $n$, $f$, and $\gamma$ are defined in Section~\ref{sec:model}. Here, $p$ denotes the number of reader threads executing concurrently within a phase; the writer is excluded since it participates in every phase. Parameter $\beta$, with $1 \le \beta \le \min\{p,f\}$, denotes the number of Byzantine readers present in a \emph{Byzantine phase} (a phase containing at least one); it captures how the adversary distributes its $f$ Byzantine readers across the phases of a cycle.\\
\noindent\textbf{Execution Bias Model.}
Under an unbiased scheduler, each phase is scheduled exactly once during a cycle, whereas a biased scheduler grants Byzantine phases additional execution opportunities, captured by the amplification factor $\gamma$. Independently, an \emph{unbiased} register linearises the up to $\beta+1$ concurrent write attempts of a phase without preference, whereas a \emph{biased} register commits a Byzantine value whenever one is attempted.

\begin{table}[h]
\centering
\small
\scalebox{0.85}{
\begin{tabular}{cc|cc|}
\cline{3-4}
                                                 &        & \multicolumn{2}{c|}{Register}      \\ \cline{3-4} 
                                                 &        & \multicolumn{1}{c|}{Biased} & Unbiased \\ \hline
\multicolumn{1}{|c|}{\multirow{2}{*}{Scheduler}} & Biased   & \multicolumn{1}{c|}{$ \frac{(f \gamma/\beta)}{(n/p)+((\gamma-1) f/\beta)} $}    & $\frac{1}{\beta+1} \cdot \frac{f \gamma}{(n/p)+((\gamma-1) f/\beta)}  $      \\ \cline{2-4} 
\multicolumn{1}{|c|}{}                           & Unbiased & \multicolumn{1}{c|}{$\frac{fp}{n\beta}$}    & $\frac{fp}{n(\beta+1)}$      \\ \hline
\end{tabular}}
\caption{Summary of the \emph{\resfac}, $P(\bot)$ for CAReg.}
\label{tab:probsummary}
\end{table}
Table~\ref{tab:probsummary} summarizes $P(\bot)$ under the four combinations of (un)biased register arbitration and (un)biased scheduling. As $f$ grows relative to $n$, adversarial writes obtain more opportunities to interfere with authenticated values, and increasing $\gamma$ lets Byzantine phases execute more frequently; concentrating Byzantine readers (larger $\beta$) instead raises the per-phase failure probability but reduces the number of Byzantine phases, and the net effect favours the adversary spreading out ($\beta = 1$). Hence, although the underlying atomic register remains linearizable and wait-free, adversarial scheduling and contention may still probabilistically suppress the visibility of legitimate values; the detailed derivations appear in Appendix~\ref{sec:AnalysisCAReg}.

\subsection{Discussion and Limitations}
\label{subsec:CAReg_discussion}
\label{subsec:baseline_limitation}

Table~\ref{tab:summaryofall} (page~\pageref{tab:summaryofall}) summarizes the properties of the baseline CAReg construction. It relies on a single atomic register and therefore requires only constant memory. Authentication is enforced externally through signed values, so the register itself remains a standard crash-tolerant linearizable object. As reflected in the ``Operation spans multiple phases'' column of Table~\ref{tab:summaryofall}, a warp may begin an operation in one phase, be suspended, and resume in a later phase. Consequently, invocation and completion need not occur within the same phase.

Retaining only a single value, however, imposes two limitations. First, since Byzantine readers may overwrite the shared location arbitrarily, a freshly written authenticated value may become temporarily unobservable. Under contention, a correct reader may return $\bot$ even when the writer is correct while both the scheduler and register are unbiased. Therefore, $P(\bot)=\frac{fp}{n(\beta+1)}$,
remains strictly positive and grows with the number of Byzantine processes. Second, a single overwritable cell provides no defense against \emph{replay}. A reader sees only the current contents of M and has no newer authenticated value for comparison, so it cannot distinguish the writer's latest write from a Byzantine re-installation of an older, validly signed pair. Such a replay preserves Byzantine linearizability. It linearises as a Byzantine write but may cause a correct reader to return a stale value.

Both limitations arise from retaining only a single value. This observation motivates constructions that maintain a bounded window of recent authenticated writes. Since the designated writer contributes one authenticated value in every phase while at most $p$ Byzantine insertions can occur within the same phase, examining the most recent $p+1$ entries guarantees that a correct reader encounters at least one valid authenticated value. Consequently, $P(\bot)=0$, and the resilience increases to the optimal bound $f \le n-1$, while preserving Byzantine-linearizable semantics.

%% file: CirclularQue.tex

The CAReg-based construction relies on a single shared location that
is repeatedly overwritten. Although this design preserves Byzantine
linearizability, contention between the writer and Byzantine
processes may still cause correct readers to return $\bot$ with
non-zero probability.


To address this limitation, we consider an alternative design based
on an ordered object that retains multiple recent values instead of
overwriting a single register.
In the remainder of this section, we assume a fair (unbiased) scheduler,
so that no warp is systematically favoured and each process is scheduled
regularly.
Rather than allowing values to be
immediately erased, they accumulate in an ordered structure, and
readers examine only the most recent entries to determine the latest
valid value. We model this structure abstractly as a \snapqueue, an ordered object with snapshot capability. The SWMR register is implemented on top of this abstraction.






\noindent\textbf{Sequential Specification of \snapqueue.}
The abstract state is a sequence $S=\langle v_1,\ldots,v_m\rangle$, initially empty. The register's initial value is \texttt{NULL} (no write yet); we distinguish \texttt{NULL} from $\bot$, which a read returns when no valid signed value is found. \snapqueue exposes:

\noindent {\bf \textsc{Insert}$(v)$.} Appends $v$ to $S$ ($S := S \cup \langle v \rangle$); returns nothing. All inserts are linearized and appear in $S$ in linearization order.

\noindent {\bf \textsc{ReadLatest}().} Returns the most recent $p+1$ values of $S$, oldest to newest (or all of $S$ if $|S|<p+1$); does not modify~$S$.
\vspace{-0.8em}
\subsection{Design Features}
\label{subsec:design}
Besides the four assumptions of the previous section, we additionally assume the following:


\noindent\textbf{Single insert call per process per phase.}
We assume that each phase is short enough that a process can invoke at most one \texttt{insert} operation on the \snapqueue. Consequently, the total number of values that can be written to the shared data structure within a single phase is bounded by $p+1$. This assumption simplifies the presentation and analysis without  losing  generality: the algorithms described below also work when phases are long enough to allow multiple insert operations by the same process. However, in that case, the size of the queue and the maximum number of values that can be stored within a single phase must be increased accordingly.

This is a reasonable assumption in our setting. The duration of a phase is determined by the scheduler and, in practical systems, typically spans only a small number of instructions. Combined with the synchronous nature of our model, this yields a natural and enforceable upper bound on the number of writes per process per phase. Adopting such a bound also allows us to avoid relying on sticky objects for Byzantine-tolerant register construction, which would otherwise be required in the presence of unbounded concurrent writes by Byzantine processes~\cite{10.1007/s00446-005-0125-8}.
\vspace{-0.5em}
\label{sec:Snapobj}
\input{snapObject}
\vspace{-0.5em}
\subsection{Discussion}
Table~\ref{tab:summaryofall} (page~\pageref{tab:summaryofall})summarizes the four realizations of \snapqueue. All four share the same guarantee, for the same reason: each phase carries one authenticated writer value and at most $px$ Byzantine insertions, so a reader examining the $px+1$ entries always observes the writer's value ($p+1$ under the single-insert assumption of Section~\ref{subsec:design}), and the sequence number bound to each value by the writer's signature identifies the most recent one among them. Hence, unlike CAReg, all four attain $f \le n-1$ and $P(\bot)=0$; and since a replayed writer-signed entry retains its old sequence number and is never selected, their freshness properties hold against the general adversary. They differ only in the cost of these guarantees, along three axes: progress, memory, and whether operations may span phases.
The \emph{lock-free linked list} caps memory at $\Delta = O(n)$ nodes via reclamation, reads the $p+1$ nodes nearest the head, and allows operations to span phases. Its only limitation is progress: insertion is a CAS retry loop, so a scheduler biased toward Byzantine processes may delay the writer indefinitely.
The \emph{\wf\ circular array} removes this limitation: an unconditional fetch-and-increment orders insertions into a fixed ring of $K$ slots, so every operation completes in a bounded number of its own steps, with reads scanning the $px+1$ newest tickets, while memory stays bounded and operations may span phases.
The remaining two \wf\ realizations occupy the extremes of the memory axis: the \emph{atomic array} achieves constant memory with a fixed $(p{+}1)$-slot window, but is correct only when operations complete within one phase; the \emph{Herlihy-Wing-inspired queue with snapshot} is the simplest to reason about, but its memory is unbounded and its reads traverse a stabilization window, costing $\Delta+p+1$.
Overall, \snapqueue\ makes the baseline's probabilistic guarantee deterministic and replay-resistant. The circular array is the choice when \wfdom\ and bounded memory are both required, the atomic array when memory is the binding constraint and operations fit in a phase, the Herlihy-Wing-inspired queue when simplicity outweighs cost, and the linked list serves as the lock-free baseline.

%% file: snapObject.tex
\subsection{Algorithmic Solution}
\label{subsec:algo-solution}
Algorithm \ref{alg:algo2} describes the phase-based write and read operations used by the register
construction in terms of the above-mentioned abstract ordered object, \snapqueue.
The concrete representation of \snapqueue and the manner in which it stores values
are intentionally left unspecified at this level and are described separately in
Appendix \ref{sec:extra_snapobject}.

\begin{algorithm}[h!]
\caption{Phase-Based Writer and Reader Operations on \snapqueue}
\label{alg:algo2}
\footnotesize
\algrenewcommand\algorithmicindent{0.8em}

\textbf{Shared variable:} \snapqueue $O$, entries initially NULL\\
\textbf{Writer-local:} $\mathit{seq} \gets 0$

\begin{minipage}[t]{0.43\columnwidth}
\begin{algorithmic}[1]
\Procedure{Write}{$v$}
    \State $\mathit{seq} \gets \mathit{seq} + 1$
    \State $sig \gets \textsc{Sign}(\langle v, \mathit{seq} \rangle)$
    \State $O.\textsc{Insert}(\langle v, \mathit{seq}, sig \rangle)$
\EndProcedure
\end{algorithmic}
\end{minipage}
\hfill
\begin{minipage}[t]{0.55\columnwidth}
\begin{algorithmic}[1]
\Procedure{Read}{}
    \State $S \gets O.\textsc{\readP}()$
    \State $valid \gets \{\bot\}$
    \ForAll{$\langle v, \mathit{seq}, sig \rangle \in S$}
        \If{$\textsc{Verify}(\langle v, \mathit{seq} \rangle, sig)$}
            \State $valid \gets valid \cup \{\langle v, \mathit{seq} \rangle\}$
        \EndIf
    \EndFor
    \State \Return $\textsc{Latest}(valid)$
\EndProcedure
\end{algorithmic}
\end{minipage}

\end{algorithm}
A write operation executed by the writer invokes \textsc{Write}$(v)$.
The writer increments a local sequence number $\mathit{seq}$, signs the pair
$\langle v,\mathit{seq}\rangle$, and inserts the signed triple
$\langle v,\mathit{seq},\mathit{sig}\rangle$ into \snapqueue, appending a
logically new element to the ordered object.
The signature helps readers to locally distinguish values produced by the
correct writer from Byzantine ones; the sequence number it binds to the value
serves only to defeat replay, since a Byzantine process can re-insert an older
writer-signed entry but cannot attach a fresh sequence number to it.
A read operation invokes \textsc{ReadLatest} on \snapqueue, which returns the $p+1$ most recently inserted entries.
The reader verifies each signature over $\langle v,\mathit{seq}\rangle$,
discards entries that fail, and returns the most recent valid value; when
several writer-signed values appear in the window, the largest sequence number
identifies the most recent.
If no valid entry is found, the reader returns $\bot$.\\
We realize \snapqueue\ in two ways: a \emph{lock-free} solution based on a
linked list, and a \emph{\wf} approach for which we give three
implementations a circular array, a Herlihy-Wing-inspired queue with
snapshot, and a fixed-size atomic array. Each realizes the same abstract
interface while exposing different trade-offs in memory consumption and
robustness under adversarial scheduling. The linked list and the circular
array are detailed below; the remaining two \wf\ implementations are deferred
to Appendix~\ref{sec:extra_snapobject}.

\subsection{Lock-Free Implementation}
\label{subsec:ll-atomic-object}

\subsubsection{\textbf{Linked-List Based Atomic Object.}}
Algorithm~\ref{alg:LinkedList} presents a lock-free implementation
of the ordered object using a singly linked list. The object
maintains a shared pointer \texttt{head} to the most recently
inserted node. Each node stores
the fields $\langle \mathit{val}, \mathit{seq}, \mathit{sig}, \mathit{next} \rangle$,
where \texttt{val} contains the stored value, \texttt{seq} the writer's
sequence number, \texttt{sig} the signature over
$\langle \mathit{val}, \mathit{seq} \rangle$,
and \texttt{next} points to the previous head, thereby maintaining
elements in reverse chronological order. The object treats \texttt{seq} and
\texttt{sig} as opaque data supplied by the writer and reader operations
(Algorithm~\ref{alg:algo2}) and performs no
authentication itself.

\noindent
\textit{\em Insert Operation.}
The \textsc{Insert}$(\langle v, \mathit{seq}, sig \rangle)$ operation creates
a new node containing the signed triple $\langle v, \mathit{seq}, sig \rangle$
and performs a lock-free insertion
at the head of the list using a compare-and-swap (CAS) primitive. The process repeatedly: (i) reads the current head, (ii) sets the new node’s \texttt{next} pointer to the observed head, and (iii) attempts $\texttt{CAS}(\textit{head}, h, n)$.
\begin{algorithm}[h!]
\caption{Linked-List-Based Atomic Object}
\label{alg:LinkedList}
\small

\textbf{Node structure:}
$\langle \textit{val}, \textit{seq}, \textit{sig}, \textit{next} \rangle$

\textbf{Shared variable:} $Node* \ head \gets NULL$

\begin{minipage}[t]{0.48\columnwidth}
\begin{algorithmic}[1]
\Procedure{Insert}{$\langle v, \mathit{seq}, sig \rangle$}
    \State $n \gets \textsc{NewNode}()$
    \State $n.\textit{val} \gets v$
    \State $n.\textit{seq} \gets \mathit{seq}$
    \State $n.\textit{sig} \gets sig$
    \State $h \gets head.\textsc{read}()$
    \State $n.\textit{next} \gets h$
    \While{$\neg head.\textsc{CAS}(h,n)$}
        \State $h \gets head.\textsc{read}()$
        \State $n.\textit{next} \gets h$
    \EndWhile
\EndProcedure
\end{algorithmic}
\end{minipage}
\hfill
\begin{minipage}[t]{0.48\columnwidth}
\begin{algorithmic}[1]
\Procedure{ReadLatest}{}
    \State $S \gets [\ ]$
    \State $h \gets head.\textsc{read}()$
    \While{$(h \neq NULL \land S.\textsc{length}() \leq p)$}
        \State $S.\textsc{append}($
        \Statex \hspace{\algorithmicindent}$\langle h.\textit{val}, h.\textit{seq}, h.\textit{sig}\rangle)$
        \State $h \gets h.\textit{next}$
    \EndWhile
    \State \Return $S$
\EndProcedure
\end{algorithmic}
\end{minipage}

\end{algorithm}
If the CAS fails due to concurrent insertions, the procedure retries with the updated head value. The successful CAS constitutes the linearization point of the operation. This guarantees lock-freedom: although an individual process may retry, some concurrent insertion succeeds in a finite number of steps.

\noindent\textit{\em ReadLatest Operation.}
The \textsc{ReadLatest} operation traverses the list
starting from \texttt{head} and collects the first $p+1$
entries encountered, newest to oldest in list order. Since new nodes are always inserted
at the head and existing links are never modified,
the traversal observes a prefix of the list corresponding
to the most recent insertions; the reader (Algorithm~\ref{alg:algo2})
then verifies the signature of each entry over
$\langle \mathit{val}, \mathit{seq} \rangle$ and returns the most recent
valid value, identified by the largest sequence number.
The operation is wait-free with respect to the parameter $p$, as it
traverses at most $p+1$ nodes.

\noindent\textbf{Correctness.}
To establish that the proposed SWMR register implementation is Byzantine-linearizable, it suffices to show that:\\
1. The implementation satisfies the sequential specification of the register, namely that every read returns the value written by the most recent preceding write in the induced sequential execution.\\
 2. Every concurrent execution history generated by the implementation is linearizable. By \defref{iditDef}, if the history observed by correct processes can be extended to a linearizable history consistent with the sequential specification of the register, then the implementation is Byzantine-linearizable.\\
The linearizability of Algorithm~\ref{alg:LinkedList}, the termination guarantees of its operations, and the \blbty\ of the SWMR register built using the linked-list-based atomic object are established in Appendix~\ref{subsec:ll-correctness}.
\noindent\textbf{Safe Memory Reclamation.}
\label{subsec:safemomoryreclamation}
Since nodes are never removed, the list would grow without bound. An optional procedure (Algorithm~\ref{alg:memoryReclaim}) detaches the suffix beyond the $\Delta$-th node with a single CAS (Figure~\ref{fig:reclamation}). Retaining $\Delta={(p+1)(\delta-1)\frac{n}{p}}=O(n)$ nodes ensures that no node still in use is reclaimed; the derivation is in Appendix~\ref{app:reclamation}.
\begin{figure} [h]
    \centering
    \includegraphics[scale=0.41]{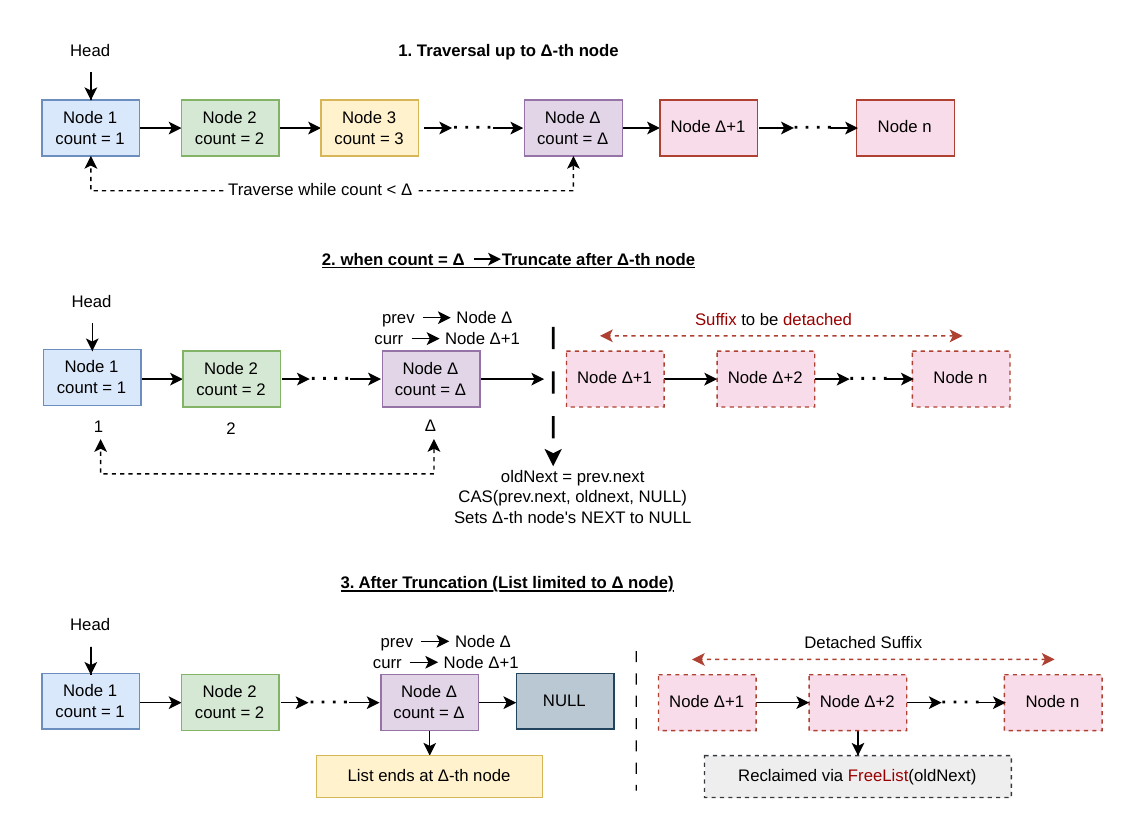}
    \vspace{-1.5em}
    \caption{Memory reclamation for the linked-list–based object with bound $(p+1)(\delta-1)\frac{n}{p}$. The list is traversed until the $\Delta$-th node, after which a CAS operation sets the $\Delta$-th node’s \texttt{next} pointer to \texttt{NULL}, detaching and reclaiming the suffix via \texttt{FreeList(oldNext)}.}

    \label{fig:reclamation}
\end{figure}





\subsection{Wait-Free Implementations}
\label{subsec:wf-implementations}
We give three \wf\ realizations of \snapqueue, all of which replace the CAS
retry loop of the linked list with an unconditional fetch-and-increment. We
detail the circular-array based object here; the remaining two are described
at the end of this section and detailed in the Appendix.
\subsubsection{\textbf{Wait-Free Circular-Array Based Atomic Object.}}
Algorithm~\ref{alg:wfarray} realizes the ordered object with a fixed circular
array. It maintains an atomic counter \texttt{tail} (a ticket dispenser) and a
shared array $A[0\ldots K-1]$; each slot holds $\langle x, t\rangle$, where
$x=\langle v,\mathit{seq},\mathit{sig}\rangle$ is the signed value and $t$ its ticket. A
position $i$ holds a live entry only for tickets congruent to $i$ modulo $K$.

\noindent
\textit{\em Insert Operation.}
\textsc{Insert}$(x)$ (i) fetch-and-increments \texttt{tail} to obtain a ticket
$t$, and (ii) writes $\langle x, t\rangle$ to $A[t \bmod K]$. Both steps are
unconditional, so every insert completes in two steps with no CAS-retry loop;
the store is its linearization point, giving wait-freedom.

\noindent
\textit{\em ReadLatest Operation.}
\textsc{ReadLatest} reads \texttt{tail} to get the current end $T$ and scans
the $W$ newest tickets ($T-1$ down to $T-W$), keeping an entry only if
$\mathit{tk}=T-a$, which rejects stale entries from an earlier lap. As a phase
issues at most $px+1$ tickets ($px$ Byzantine, one writer), $W=px+1$ keeps the
writer's latest value in the window ($p+1$ when $x=1$). It performs $W$ reads,
terminating in $W+1$ steps, hence is wait-free.
\begin{algorithm}[h]
\caption{Wait-Free Circular-Array Based Atomic Object (WF-Array)}
\label{alg:wfarray}
\small
\begin{algorithmic}[1]
\Statex
\State \textbf{Shared variables:}
\State \quad atomic$<$int$>$ $\mathit{tail} \gets 0$ \Comment{ticket dispenser}
\State \quad atomic$<$Entry$>$ $A[0 \ldots K-1] \gets \texttt{NULL}$ \Comment{Entry $=\langle \mathit{val},\mathit{seq},\mathit{sig}\rangle$}
\Procedure{Insert}{$\langle v,\mathit{seq},\mathit{sig}\rangle$} \Comment{any process may invoke}
    \State $t \gets \mathit{tail}.\textsc{Fetch-And-Increment}()$ \label{lin:wffai} \Comment{unconditional}
    \State $A[t \bmod K] \gets \langle v,\mathit{seq},\mathit{sig}\rangle$ \label{lin:wfstore} \Comment{store into reserved slot}
\EndProcedure
\Procedure{\readP}{$ $} \Comment{any process may invoke}
    \State $S \gets [\ ]$
    \State $T \gets \mathit{tail}.\textsc{read}()$ \label{lin:wftail}
    \For{$a = 1$ \textbf{to} $W$} \Comment{newest to oldest}
        \State $\langle \mathit{val},\mathit{seq},\mathit{sig}\rangle \gets A[(T-a) \bmod K]$ \label{lin:wfload}
        \If{$\langle \mathit{val},\mathit{seq},\mathit{sig}\rangle \neq \texttt{NULL} \;\land\; \mathit{tkt}=T-a$} \label{lin:wfmatch}
            \State $S.\textsc{append}(\langle \mathit{val},\mathit{seq},\mathit{sig}\rangle)$
        \EndIf
    \EndFor
    \State \Return $S$ \Comment{entries newest first; reader verifies and selects}
\EndProcedure
\end{algorithmic}
\end{algorithm}


\noindent\textbf{\em Assumptions.}
Beyond the assumptions of Sections~\ref{subsec:CARegAssumption}
and~\ref{subsec:design}, the analysis assumes that within a phase the insert
group executes before the read group, that a process completes at most $x$
inserts per scheduled phase, round-robin scheduling, and that a pending store
lands within a bounded number of scheduled phases of its fetch-and-increment;
the precise statements (R1)--(R4), together with the choice of the window $W$
and the array size $K$, appear in Appendix~\ref{subsubsec:wfarray-object}.

\noindent\textit{\em Correctness.}
The \wfdom\ and linearization points of Algorithm~\ref{alg:wfarray}, its
Byzantine ordering properties with respect to Definition~\ref{def:byzlinz},
and the resulting \blbty\ and resilience of the SWMR register built on it are
established in Appendix~\ref{subsubsec:wfarray-object}.

\medskip
\noindent
We give two further wait-free realizations of the \snapqueue, both using an
unconditional fetch-and-increment, with full descriptions and proofs in the
appendix. \\
The \textbf{\hwq} (Algorithm~\ref{alg:HWQueue}), inspired by the
Herlihy--Wing queue, fetch-and-increments a counter \texttt{tail} to order
inserts and scans the most recent entries on read, using unbounded memory
(Appendices~\ref{subsec:hwqueue-snapshotalgo}).\\ 
The \textbf{Array-Based Atomic Object}
(Algorithm~\ref{alg:ArrayObject}) stores the window in a fixed $(p{+}1)$-slot
array indexed modulo $p+1$, using constant memory but correct only within a
single phase (Appendices~\ref{subsubsec:array-atomic-object}).

\noindent\textbf{Resilience of the \snapqueue Construction}
Let $W$ denote the writer and $R = \{r_1,\dots,r_n\}$ the set of readers.
Let $F \subseteq R$ be the set of Byzantine readers with $|F| = f$,
and let $C = R \setminus F$ denote the set of correct readers.

\begin{theorem}
If $W$ is correct, then for every execution and every $r_i \in C$, $P(r_i \text{ returns } \bot) = 0$.
\end{theorem}

\begin{theorem}
If $W$ is correct and $C \neq \emptyset$,
then the SnapObject-based construction
is Byzantine-linearizable for all $f \le n - 1$.(Detailed proofs are given in Appendix~\ref{subsubsec:resilienceofsnapqueue}.)
\end{theorem}

\begin{table*}[h!]
\centering
\small
\scalebox{0.71}{
\begin{tabular}{|c|c|c|c|c|c|c|c|c|c|}
\hline
\textbf{Type}
& \textbf{\begin{tabular}[c]{@{}c@{}}Underlying \\ Implementation\end{tabular}}
& \textbf{\begin{tabular}[c]{@{}c@{}}Progress \\ Guarantee\end{tabular}}
& \textbf{Memory}
& \textbf{$\mathbf{f}$-resilience}
& \textbf{$P(\bot)$}
& \textbf{Adversary}
& \textbf{\begin{tabular}[c]{@{}c@{}}Signature\\ Aware\end{tabular}}
& \textbf{\begin{tabular}[c]{@{}c@{}}No.\ of reads\\ in a phase\end{tabular}}
& \textbf{\begin{tabular}[c]{@{}c@{}}Operation spans\\ multiple phases \end{tabular}} \\
\hline

Crash-Tolerant
& Atomic Register
& Wait-free
& Constant
& $\dfrac{(\beta+1)qn}{p}$
& $\dfrac{fp}{n(\beta+1)}$\footnotemark
& Non-replaying
& NA
& $1$
& Yes \\
\hline

\multirow{4}{*}{SnapObject}
& \begin{tabular}[c]{@{}c@{}}Linked List\end{tabular}
& Lock-free
& Bounded
& $n-1$
& $0$
& General
& NA
& $p+1$
& Yes \\ \cline{2-10}

& \begin{tabular}[c]{@{}c@{}}Circular Array\end{tabular}
& Wait-free
& Bounded
& $n-1$
& $0$
& General
& NA
& $px+1$
& Yes \\ \cline{2-10}

& \begin{tabular}[c]{@{}c@{}}Atomic Array\end{tabular}
& Wait-free
& Constant
& $n-1$
& $0$
& General
& NA
& $p+1$
& No \\ \cline{2-10}

& \begin{tabular}[c]{@{}c@{}}\hwq \\ \end{tabular}
& Wait-free
& Unbounded
& $n-1$
& $0$
& General
& NA
& $\Delta+p+1$
& Yes \\
\hline

\begin{tabular}[c]{@{}c@{}}Idealized Signature-Aware Memory\end{tabular}
& \begin{tabular}[c]{@{}c@{}}Atomic Register\end{tabular}
& Wait-free
& Constant
& $n-1$
& $0$
& General
& Yes
& $1$
& Yes \\
\hline
\end{tabular}}

\caption{Design trade-offs and resilience properties of Byzantine-linearizable SWMR register constructions.}
\label{tab:summaryofall}

\end{table*}

%% file: conclusion.tex
\section{Idealized Signature-Aware Memory}
\label{sec:ideal}
The previous constructions enforce authenticity externally: processes first retrieve values from the shared object and then verify signatures before accepting the returned value. We now consider a stronger shared-memory abstraction in which authentication is incorporated directly into the memory object itself. In this design, the shared object stores the writer's public key and internally verifies the supplied signature before accepting a write operation; the register state is thus updated only if the write carries a valid signature from the writer.
Since Byzantine processes cannot generate valid signatures, unauthorised writes are rejected before modifying the register state. As shown in Table~\ref{tab:summaryofall}, this eliminates Byzantine overwrites and allows the construction to achieve maximal $f$-resilience $(n-1)$ and $P(\bot)=0$. The construction additionally preserves constant memory and requires only a single read per phase because it still relies on a single atomic memory location. Moreover, unlike CAReg, the signature-aware memory is robust against the \emph{general} adversary of Section~\ref{sec:model}: just as the object permanently holds the writer's public key, it also remembers the sequence number of the last accepted write. The writer signs the pair $\langle v, seq \rangle$ (as in Algorithm~2), so the counter is cryptographically bound to the value; the memory accepts a write only if the signature verifies \emph{and} the supplied sequence number strictly exceeds the stored one. A Byzantine process replaying an older, validly signed pair cannot alter its bound counter without invalidating the signature, and the stale counter causes the write to be rejected. Replays therefore have no effect on the register state.
However, these guarantees rely on a significantly stronger assumption. Embedding signature verification directly into shared-memory objects is complex and implementation-specific. We discuss possible implementation considerations in Appendix~\ref{sec:sign}. Consequently, the construction primarily serves to characterize the upper boundary of the design space rather than a directly deployable implementation.
\section{Conclusion and Future Work}
\label{sec:conclusion}
This paper studied the construction of Byzantine-linearizable shared-memory objects in a GPU-inspired execution model where threads freely share on-chip memory and Byzantine processes may overwrite arbitrary locations. Focusing on the SWMR register, we investigated whether atomic semantics can still be preserved for correct processes under unrestricted shared-memory access.

We first formalised a phase-aware notion of Byzantine linearizability tailored to lock-step GPU execution and then presented three register constructions occupying different points in the design space. The crash-tolerant atomic-register construction achieves constant memory and a single read per phase, but provides only probabilistic resilience. In contrast, the \snapqueue-based constructions retain recent authenticated writes and achieve deterministic resilience guarantees with $P(\bot)=0$, at the cost of higher memory and synchronisation overheads. The signature-aware memory abstraction further combines constant memory with deterministic resilience under stronger hardware assumptions.

The constructions also differ in how well they resist replay. CAReg keeps a single value, so a reader has nothing to compare it against and its freshness guarantees hold only against a non-replaying adversary. The \snapqueue-based and signature-aware constructions bind a sequence number to each signed value, so they remain correct against the general adversary.

Overall, our results show that GPU execution properties such as lock-step warp execution, bounded instruction phases, and structured scheduling can partially compensate for the absence of traditional access-control mechanisms in Byzantine shared-memory systems.

Several directions remain open for future work. An immediate extension is the construction of richer Byzantine-tolerant shared-memory objects, including stacks, queues, and snapshots, using the SWMR registers developed here. Another direction is to study the synchronization power of objects in the GPU-inspired model and develop a Byzantine consensus hierarchy analogous to consensus-number hierarchies in classical shared memory. 
Finally, an intriguing direction is to extend the model beyond a single SM and take into consideration recent advances on synchronization mechanisms across SMs on GPUs. Understanding Byzantine linearizability and resilience in such multi-SM environments is challenging.

%% file: ByzLin.tex
Using the definitions of \blbty as a foundation, we now formulate a notion of Byzantine linearizability suitable for the GPU-inspired, phase-based execution model. Our definition preserves the core guarantees of the original formulations while accounting for synchronous execution within phases and explicitly addressing the behavior of Byzantine writers. If the lock-step execution assumption is removed, the additional constraints imposed by our definition disappear, and our definition coincides with the general \blbty definition of Cohen and Keidar~\cite{CKDISC2021}.\\
We construct a \blble SWMR register using shared memory. For clarity and consistency, we use the notation summarized in Table~\ref{tab:notation} throughout the paper.


\noindent\textbf{Uniqueness of Written Values.}
For simplicity, we assume that every value written by the correct writer is unique.
This assumption can be enforced in practice by the use of sequence numbers.
 \begin{definition}[Byzantine Linearizability for SWMR Registers in the GPU-Inspired Model] The following conditions must hold for correct processes:
\label{def:byzlinz}

    \noindent\textbf{$\mathbf{(a)}$ No Future Read.}
If a read operation completes in phase $j$, then the value it returns must have been written
by a write operation that completes in phase $j$ or earlier.
 Formally,
	\[
	(\forall t \in C,\,i,j:\ 
	\Bigl(\Readij{i}{j}{t}{}{x}\Bigr)
	\ \Longrightarrow\
	\Bigl( \exists\  \writeij{p}{q}{}{x}: p \leq j \Bigr)
	)
	\]

    \noindent\textbf{$\mathbf{(b)}$ No Read from Distant-Past.}
 A read operation must never return a value that has already been superseded by a more recent completed write before the phase in which the read occurs. Formally, if a value $x$ is written by a writer that completes in phase $i$, this value is later overwritten by a write that starts in phase $j$ with $i<j$ and completes in phase $k$, and a read by thread $t$ occurs in a strictly later phase $l$, then the read cannot return the old value $x$\footnote{Here even if the second writer wrote the same value as previous writer i.e., $x=y$, the reader should return the second write operation associated value only.}.
    \[
    \forall\, i, j, k, l, \text{ and } \forall t \in C:\;
    \bigl(i < j < k < l \text{ s.t. }  \]
     \[ \writeij{*}{i}{}{x} \;\rightarrow\; 
  \writeij{j}{k}{}{y} \; \rightarrow\;\Readij{l}{*}{t}{}{z}\bigr)
    \;\Longrightarrow\; z \neq x .
    \]

    \noindent \textbf{$\mathbf{(c)}$ No New--Old Inversion.}
A later read must not return a value originating from an older write than the one observed
by an earlier read.
Let the writer perform write operations sequentially, writing distinct values
$x_1, x_2, x_3, \ldots$.
Let $\Readthreeij{i}{j}{t}$ and $\Readthreeij{k}{l}{t'}$ be two read operations executed by correct threads such that
$j < k$, and suppose that
$\Readij{i}{j}{t}{a}{x}$
and
$\Readij{k}{l}{t'}{b}{x} .$
The execution satisfies \emph{No New--Old Inversion} if
$b \ge a.$

\noindent\textbf{$\mathbf{(d)}$ Intra-Phase Reader Consistency.}

Recall that from the model, within the same phase all the processes executing the same instruction run in lock-step manner, as illustrated in Figure~\ref{fig:phase-lockstep}.
 This means all the correct readers completing their read operation in the same phase must return the same value. In other words,

\[
\forall\ \{1,2,...p\} \in C: \text{  } \Readij{i}{j}{t}{1}{x}, \Readij{i}{j}{t}{2}{x}, \ldots, \Readij{i}{j}{t}{p}{x}\]\[ \;\Longrightarrow\; x_1 = x_2 = \cdots = x_p.
\]
\end{definition}

\noindent
\textbf{\em Remark.}
All the properties we have mentioned above hold true even when start and end phase of each operations is the same i.e., the function call completes within a single phase. Furthermore, these properties hold regardless of the writer being Byzantine or not, as long as the written values are distinguishable\footnote{In the case that the writer is Byzantine, the written value might be useless with respect to the application, but this does not violate the ordering of the events.}.

%% file: CARegCorrectness.tex
\label{sec:correctness}

We show that Algorithm~\ref{alg:SWMR_AReg} implements a Byzantine-linearizable single-writer multi-reader
(SWMR) register according to Definition~\ref{def:byzlinz}.

\paragraph{\textbf{Proof Structure.}}
The proof proceeds in three steps.
First, we establish that the underlying crash-tolerant atomic register (CAReg) remains
linearizable even when method invocations are issued by Byzantine processes.
Second, we show that any value returned by a correct reader (other than $\bot$) is written
by the writer.
Finally, we prove that the register satisfies the ordering constraints required by Byzantine
linearizability: \emph{No Future Read}, \emph{No Read from Distant-Past Value},
\emph{No New--Old Inversion}, and \emph{Intra-Phase Reader Consistency}.
Together, these properties imply Byzantine linearizability.


\subsubsection{\textbf{Object-Level Linearizability}}
\label{subsec:object-linearizability}

\paragraph{\textbf{Crash-Tolerant Wait-Free Objects.}}
An abstract shared object $O$ is \emph{crash-tolerant wait-free} if, in executions with crash
failures, it satisfies the following properties:

\begin{itemize}
  \item \textbf{Wait-Freedom.}
  Every method invocation by a correct process completes in a finite number of its own steps,
  regardless of the behavior (including crashes) of other processes.

  \item \textbf{Crash-Tolerant Correctness.}
  The object remains linearizable with respect to its sequential specification in all
  executions in which processes may crash but do not exhibit Byzantine behavior.
\end{itemize}

Crash-faulty processes may halt prematurely and fail to respond, but they do not perform arbitrary or adversarial actions and cannot corrupt the object state.


\begin{lemma}[Linearization Point Localization]
\label{lem:lp-localization}
Let $op$ be a completed method invocation that starts in phase $i$ and completes in phase
$j$, where $i \le j$.
Even if the execution of $op$ spans multiple phases, there exists a unique linearization
point
\[
\ell(op) \in \{i, i+1, \ldots, j\},
\]
and this linearization point occurs within a single phase between invocation and response.
\end{lemma}

\begin{proof}
By linearizability, every completed method invocation admits a unique instantaneous
linearization point between its invocation and response.
Execution is partitioned into totally ordered phases, and time within each phase is totally
ordered.
Hence, any instantaneous event, including a linearization point, must belong to exactly one
phase.
Since $op$ is invoked in phase $i$ and completes in phase $j$, its linearization point occurs
in some phase $k$ with $i \le k \le j$.
\end{proof}


\begin{lemma}[Invocation-Independence of Linearizability]
\label{lem:invocation-independence}
Let $O$ be a crash-tolerant linearizable object.
For any execution of $O$, the history consisting of all \emph{completed} method invocations
is linearizable, regardless of whether these invocations are issued by correct or Byzantine
processes.
\end{lemma}


\begin{proof}
In the system model (Section~\ref{sec:model}), processes interact with shared memory exclusively through
method invocations on the object $O$.
Although a Byzantine process may invoke methods arbitrarily or halt during an invocation,
it cannot modify the internal semantics of $O$.

Linearizability is defined with respect to the history of \emph{completed} method invocations.
For each completed invocation, crash-tolerant linearizability guarantees the existence of a
linearization point that lies between the invocation and its response.
Invocations that do not complete, whether due to crashes or Byzantine behavior, are simply
excluded from the history and need not be linearized.

Therefore, the history consisting of all completed method invocations is linearizable,
independent of whether these invocations are issued by correct or Byzantine processes.
\end{proof}



\begin{lemma}[CAReg Linearizability under Byzantine Processes]
\label{lem:careg-linearizable}
All completed read and write operations on the crash-tolerant atomic register CAReg are
linearizable in the Byzantine GPU-inspired shared-memory model.
\end{lemma}

\begin{proof}
CAReg is linearizable under crash failures by definition.
By Lemma~\ref{lem:invocation-independence}, linearizability is preserved even when method
invocations are issued by Byzantine processes.
\end{proof}





\subsubsection{\textbf{Value Validity}}
\label{subsec:value-validity}

\paragraph{\textbf{Authorized Writer.}}
The \emph{authorized writer} is the process whose public key is used by readers to verify
signatures attached to written values.
A value is \emph{writer-generated} if it is accompanied by a signature that verifies
successfully under the authorized writer’s public key.


\begin{lemma}[Writer Authenticity - Value Validity]
\label{lem:writer-authenticity}
If a read operation executed by a correct process returns a value
$v \neq \bot$, then $v$ was written by the authorized writer.
\end{lemma}
\begin{proof}
A reader returns a value only if the signature verifies under the authorized
writer's public key (Algorithm~\ref{alg:SWMR_AReg}). Since Byzantine
processes cannot forge valid signatures, any returned value originates from
the writer. (A Byzantine process may \emph{replay} a writer-signed value; the
replayed value is still writer-generated, so validity is not violated.)
\end{proof}
\paragraph{Observation.}
If the authorised writer itself is Byzantine, then it might share its private key with other Byzantine processes.

\subsubsection{\textbf{Byzantine linearizability}}
\label{subsec:ordering-properties}
\begin{itemize}
    \item \textbf{No Future Read.}

\begin{lemma}[No Future Read]
\label{thm:no-future}
For every correct thread $t$ and phases $i \le j$,
\[
\Readij{i}{j}{t}{}{v} \neq \bot \;\Rightarrow\; \exists \writeij{p}{q}{}{v}\text{ such that } p \le j .
\]
\end{lemma}
\begin{proof}
The read operation $\Readthreeij{i}{j}{t}$ invokes a read on the underlying atomic register
(Algorithm~\ref{alg:SWMR_AReg}, Line~7) and returns the obtained value if and only if the associated
signature verifies; otherwise, it returns $\bot$ (Lines~8--11).
By Lemma~\ref{lem:lp-localization}, the read $\Readthreeij{i}{j}{t}$ has a unique linearization
point $\ell(\Readthreeij{i}{j}{t})$ such that $i \le \ell(\Readthreeij{i}{j}{t}) \le j$, Similarly the write $W^{p,q}$ has a unique linearization
point $\ell(\writeij{p}{q}{}{v})$ such that $p \le \ell(\writetwoij{p}{q}) \le q$.

Since $\Readthreeij{i}{j}{t}$ is executed by a correct thread and returns $v \neq \bot$,
the value $v$ must have been written by some write operation $\writeij{p}{q}{}{v}$.
Assume, $p > j$.
Then,
\[
i\le \ell(\Readthreeij{i}{j}{t}) \le j < p \le \ell(\writetwoij{p}{q}) \le q,
\]
which implies that $\writeij{p}{q}{}{v}$ linearizes strictly after $\Readthreeij{i}{j}{t}$.

This contradicts the sequential specification of an atomic register, which forbids
a read from returning a value written by a write that linearizes after the read.
Therefore, $p \le j$.

\end{proof}


\item  \textbf{No Read from Distant-Past.}

\begin{lemma}[No Read from Distant-Past]
\label{thm:no-distant-past}
Assume the non-replaying adversary of Section~\ref{subsec:CARegAssumption}.
Let $\writeij{i}{j}{}{x}$ and $\writeij{k}{l}{}{y}$ be two write operations issued by the correct writer such that
$j < k$.
For any correct read operation $\Readthreeij{p}{q}{t}$ with $p > l$,
\[
\Readij{p}{q}{t}{}{z} \neq x .
\]
\end{lemma}

\begin{proof}

Let $\writeij{i}{j}{}{x}$ and $\writeij{k}{l}{}{y}$ be two write operations issued by the correct writer such
that $j < k$, and let $\Readthreeij{p}{q}{t}$ be a read operation executed by a correct thread
$t$ with $p > l$.

In Algorithm~\ref{alg:SWMR_AReg}, a read operation invokes a read on the underlying atomic register
(Line~7) and returns either the value currently stored in the register if the
signature verifies, or $\bot$ otherwise (Lines~8--11).

By Lemma~\ref{lem:lp-localization}, the linearization points satisfy
\[
i \le \ell(\writeij{i}{j}{}{x}) \le j,\quad
k \le \ell(\writeij{k}{l}{}{y}) \le l,
\]
\[
p \le \ell(\Readij{p}{q}{t}{}{z}) \le q.
\]
Since $j < k$ and $l < p$, it follows that
\[
\ell(\writetwoij{i}{j}) < \ell(\writetwoij{k}{l}) < \ell(\Readthreeij{p}{q}{t}).
\]

Thus, the write $\writeij{k}{l}{}{y}$ overwrites the value written by $\writeij{i}{j}{}{x}$ before the
read $\Readij{p}{q}{t}{}{z}$ linearizes.
By linearizability of the underlying atomic register (CAReg)
(Lemma~\ref{lem:careg-linearizable}), the value returned by $\Readthreeij{p}{q}{t}$ must
correspond to the most recent write whose linearization point precedes that of the
read. Consequently, the value $x$ written by $\writeij{i}{j}{}{x}()$ cannot be returned.
Moreover, if the register contains a value written by a Byzantine process,
Algorithm~1 returns $\bot$, not $x$.
Therefore, $\Readij{p}{q}{t}{}{z} \neq x$.

\end{proof}



\item \textbf{No New-Old Inversion.}

\begin{lemma}[No New--Old Inversion]
\label{thm:no-new-old}
Assume the non replaying adversary of Section~\ref{subsec:CARegAssumption}.
A later read must not return a value originating from an older write than the one observed
by an earlier read.
Let the correct writer perform write operations sequentially, writing distinct values
$x_1, x_2, x_3, \ldots$.
Let $\Readthreeij{i}{j}{t}$ and $\Readthreeij{k}{l}{t'}$ be two read operations executed by correct threads such that
$j < k$, and suppose that
\[
\Readij{i}{j}{t}{a}{x}
\quad\text{and}\quad
\Readij{k}{l}{t'}{b}{x} .
\]
The execution satisfies \emph{No New-Old Inversion} if
\[
b \ge a.
\]
\end{lemma}
\begin{proof}
The correct writer performs write operations sequentially and writes distinct values
$x_1, x_2, x_3, \ldots$ in increasing phase order, the corresponding write operations are
totally ordered by their linearization points. In particular, for all $u < v$,
\[
\ell(\writeij{p}{q}{u}{x}) < \ell(\writeij{r}{s}{v}{x}) .
\]

Because $j < k$, the read operation $\Readthreeij{i}{j}{t}$ completes before the read operation
$\Readthreeij{k}{l}{t'}$ begins. By linearizability,
\[
\ell(\Readthreeij{i}{j}{t}) < \ell(\Readthreeij{k}{l}{t'}) .
\]
We now assume, for contradiction, that $b < a$.
By the sequential order of the correct writer’s writing,
Since both reads are executed by correct threads and return non-$\bot$ values, each returned
value must correspond to a write whose linearization point precedes that of the read.
Thus,
\[
\ell(\writeij{p}{q}{a}{x}) < \ell(\Readthreeij{i}{j}{t})
\quad\text{and}\quad
\ell(\writeij{r}{s}{b}{x}) < \ell(\Readthreeij{k}{l}{t'}) .
\]

\[
\ell(\writeij{r}{s}{b}{x}) < \ell(\writeij{p}{q}{a}{x}) .
\]

Combining the above inequalities yields
\[
\ell(\writeij{r}{s}{b}{x}) < \ell(\writeij{p}{q}{a}{x}) < \ell(\Readthreeij{k}{l}{t'}) .
\]

Since the object is insert only with no delete and linearizable, a read must return the value written by
the most recent write whose linearization point precedes that of the read.
Therefore, $\Readthreeij{k}{l}{t'}$ cannot return $x_b$ when a later value $x_a$ has already been
written and precedes its linearization point.
This contradicts the assumption that $\Readij{k}{l}{t'}{b}{x}$.

Hence, the assumption $b < a$ is false, and it must hold that $b \ge a$.
\end{proof}

\noindent\emph{Remark (replay).} Against the general adversary of
Section~\ref{subsec:CARegAssumption}, Lemmas~\ref{thm:no-distant-past}
and~\ref{thm:no-new-old} can fail: a replayed writer-signed value passes
signature verification and may be returned even though it has been
superseded. This does not affect Byzantine linearizability
(Theorem~\ref{thm:byz-linearizable}), which rests on value validity alone.
\item \textbf{Intra-Phase Reader Consistency.}




\begin{lemma}[Intra-Phase Reader Consistency]
\label{thm:intra-phase}
Let $j$ be any phase, and let
\[
\Readthreeij{i}{j}{1}, \Readthreeij{i}{j}{2}, \ldots, \Readthreeij{i}{j}{p}
\]
be read operations executed by correct threads that all complete in phase $j$.
Then all such reads return the same value.
Formally,
\[
\Readij{i}{j}{1}{1}{x}, \Readij{i}{j}{2}{2}{x}, \ldots, \Readij{i}{j}{p}{p}{x}
\;\Longrightarrow\;
x_1 = x_2 = \cdots = x_p .
\]
\end{lemma}
\begin{proof}
Consider any correct thread $t$ whose read operation $R^t_{i,j}$ completes in phase $j$.
By Lemma~\ref{lem:lp-localization}, $\Readthreeij{i}{j}{t}$ has a unique linearization point
$\ell(\Readthreeij{i}{j}{t})$ such that $i \le \ell(\Readthreeij{i}{j}{t}) \le j$.

By the GPU-inspired system model (\secref{model}), execution within a phase follows the SIMT
lock-step discipline: all correct threads scheduled in phase $j$ execute the same instruction
sequence synchronously.
In particular, all read operations that complete in phase $j$ invoke the same read of the
shared register CAReg and perform identical, deterministic local computations.

Since CAReg is linearizable (Lemma~\ref{lem:careg-linearizable}), all reads whose
linearization points lie in phase $j$ observe the same logical state of the register.
Moreover, signature verification and value selection are deterministic functions of this
state.

Therefore, for any two correct reads $\Readthreeij{i}{j}{a}$ and $\Readthreeij{i}{j}{b}$ that complete in phase
$j$, the values they return must be identical.
Hence,
\[
x_1 = x_2 = \cdots = x_p .
\]
\end{proof}


\end{itemize}

\subsubsection{\textbf{Byzantine Linearizability}}

\begin{theorem}
\label{thm:byz-linearizable}
Algorithm~\ref{alg:SWMR_AReg} implements a Byzantine-linearizable SWMR register.
\end{theorem}

\begin{proof}
By Lemma~\ref{lem:writer-authenticity}, every non-$\bot$ value returned by a
correct read was written by the correct writer (value validity). We build a
witness for Definition~\ref{def:iditDef} directly. Let $H$ be any history of
Algorithm~\ref{alg:SWMR_AReg}. Arrange the completed operations of correct
processes in any total order that extends real-time precedence, and place the
writer's writes in that order. Now, immediately before each correct read that
returns a value $u$, if the operation just before it does not already write
$u$, insert a $\textsc{write}(u)$ attributed to a Byzantine process. By value
validity, each such $u$ is $\bot$, the initial value $v_0$, or a writer value,
so every inserted write is a legitimate Byzantine operation.

In the resulting sequential history, every read is immediately preceded by a
write of the value it returns, so the history is legal; correct operations
keep their real-time order and the inserted operations are fresh Byzantine
events, so real-time precedence is respected; and the projection onto correct
processes is unchanged. Hence $H$ is Byzantine linearizable. Under the
non-replaying adversary, Lemmas~\ref{thm:no-future}, \ref{thm:no-distant-past},
\ref{thm:no-new-old}, and~\ref{thm:intra-phase} additionally establish the
ordering conditions of Definition~\ref{def:byzlinz}.
\end{proof}

%% file: AnalysisCAReg.tex
\label{sec:AnalysisCAReg}
For simplicity, we assume that read and write operations start and complete within the same phase. Under this assumption, we analyze the probability that a correct reader returns $\bot$ when the writer is also correct. This case captures the meaningful reliability guarantees of the construction. When the writer is Byzantine, no guarantees can be made about the validity of the values written by that process, since a Byzantine writer may arbitrarily deviate from the protocol.


\paragraph{System Parameters.}
\begin{itemize}
\item One writer (correct), present in every phase.
\item \(n\) total reader processes.
\item \(f\) Byzantine readers.
\item \(p\) reader processes executed per phase.
\item \(\beta \le \min\{p, f\}\) Byzantine readers may appear in any phase.
\item \(\gamma \ge 1\): scheduler bias factor, indicating that Byzantine phases are scheduled \(\gamma\) times more frequently than correct phases.
\end{itemize}

\paragraph {\textbf{Cycles.}}
A \emph{cycle} is a finite sequence of phases in which every process in the system is scheduled
at least once. Since each phase schedules exactly $p$ distinct reader processes, every cycle consists
of at least $\lceil n / p \rceil$ phases.

\paragraph {\textbf{Correct Phases.}}
A phase is \emph{correct} if it contains no Byzantine readers. In such a phase, the correct writer
is the only process capable of producing a value with a valid signature. Consequently, a read
operation executed by a correct reader never returns $\bot$, and therefore
\[ P(\bot \mid \text{correct phase}) = 0 .
\]

\paragraph {\textbf{Byzantine Phases.}}
A phase is \emph{Byzantine} if it contains $\beta \ge 1$ Byzantine readers in addition to the correct
writer. In such a phase, at most $\beta + 1$ write attempts may occur: one issued by the writer and
up to $\beta$ issued by Byzantine readers.\\
We analyze the probability that a correct reader returns $\bot$ by considering the effects of
atomic-register arbitration and scheduler behavior. We distinguish the following cases.
\begin{enumerate}
    \item \textbf{Unbiased AReg.}
    Under an unbiased atomic register (\AReg), all write attempts have an equal chance to be the last written value before the reads. Since Byzantine processes
    cannot forge valid signatures, the probability that a correct reader
    returns \(\bot\) is
    \[
        P(\bot \mid \text{Byzantine phase, \AReg unbiased})
        = \frac{\beta}{\beta+1}.
    \]
    
    \item \textbf{Biased AReg.}
    When \AReg is fully biased in favour
of Byzantine threads (i.e., whenever any Byzantine attempt is present, \AReg
always commits a Byzantine value), we have

    \[
        P(\bot \mid \text{Byzantine phase, \AReg biased})
        = 1.
    \]
    
    \item \textbf{Unbiased Scheduler.}
Under an unbiased scheduler, scheduling is fair across a cycle.
Across the $n/p$ phases of a cycle, exactly $f/\beta$ phases are Byzantine.
Each phase is selected uniformly, and therefore the probability that a phase is
Byzantine is
\[
P(\text{Byzantine phase} \mid \text{scheduler unbiased})
= \frac{f/\beta}{n/p}
= \frac{fp}{n\beta}.
\]

\item \textbf{Biased Scheduler.}
Under a scheduler biased towards Byzantine phases by a factor $\gamma \ge 1$,
phases containing Byzantine readers are selected more frequently than correct
phases.
In particular, each Byzantine phase is selected $\gamma$ times more frequently
than a correct phase.
The effective number of Byzantine phase selections in a cycle is therefore
$f\gamma/\beta$, while the total weighted number of phase selections is
\[
\frac{n}{p} + \frac{(\gamma - 1)f}{\beta}.
\]
As a result, the probability that a selected phase is Byzantine under a biased
scheduler is
\[
P(\text{Byzantine phase} \mid \text{scheduler biased})
= \frac{(f\gamma/\beta)}{(n/p) + ((\gamma - 1)f/\beta)}.
\]



\end{enumerate}

\subsubsection{\textbf{Probabilistic Analysis of Returning $\bot$ Under \AReg and Scheduler Bias}}

We analyze the probability that a correct reader returns ~$\bot$
in a system containing Byzantine readers. 


A correct reader returns $\bot$ when the signature of the value it reads is not
valid. We refer to any phase containing no Byzantine readers as a \emph{correct
phase}.

\bigskip
We now state the four scenarios corresponding to the four combinations of
(un)biased \AReg arbitration and (un)biased scheduling.

\begin{enumerate}
    \item  \textbf{\AReg Unbiased, Scheduler Unbiased}
In a system with $n$ readers, $f$ Byzantine readers, phase size $p$, and
$\beta$ Byzantine readers per Byzantine phase, under unbiased \AReg and an unbiased
scheduler, the probability that a correct reader returns $\bot$ in a given phase
is
\[
    P(\bot) 
        = \frac{f p}{n (\beta+1)}.
\]

\begin{proof}
Using the law of total probability,
\[
\begin{aligned}
P(\bot)
 &= P(\bot \mid \text{correct phase}) \cdot P(\text{correct phase}) \\
 &\quad + P(\bot \mid \text{Byzantine phase}) \cdot
         P(\text{Byzantine phase}) \\
 &= 0 \cdot P(\text{correct phase})
      + \frac{\beta}{\beta+1}\cdot \frac{fp}{n\beta} \\
 &= \frac{fp}{n(\beta+1)}.
\end{aligned}
\]
\end{proof}
Since $P(\bot \mid \text{correct phase}) = 0$, in the following parts we will directly use
\[
P(\bot)
= P(\bot \mid \text{Byzantine phase}) \cdot
  P(\text{Byzantine phase}).
\]

\item  \textbf{\AReg Biased, Scheduler Unbiased}
Under \AReg biased in favour of Byzantine writes and an unbiased scheduler,
the probability that a correct reader returns $\bot$ is
\[
    P(\bot) 
        = \frac{fp}{n\beta}.
\]

\begin{proof}
Since $P(\bot \mid \text{Byzantine phase}) = 1$ and
$P(\text{Byzantine phase}) = \frac{fp}{n\beta}$,
\[
    P(\bot) 
        = 1 \cdot \frac{fp}{n\beta}
        = \frac{fp}{n\beta}.
\]
\end{proof}



\item  \textbf{\AReg Unbiased, Scheduler Biased}
Under unbiased \AReg and a scheduler biased towards Byzantine phases by a factor
$\gamma \ge 1$, the probability that a correct reader returns $\bot$ is
\[
    P(\bot) 
        = \frac{f \gamma}{((n/p)+((\gamma-1) f/\beta)) (\beta+1)}.
\]

\begin{proof}
Using $P(\bot \mid \text{Byzantine phase}) = \frac{\beta}{\beta+1}$ and
\[
P(\text{Byzantine phase}) = \frac{(f \gamma/\beta)}{(n/p)+((\gamma-1) f/\beta)},
\]
\[
\begin{aligned}
P(\bot)
 &= \frac{\beta}{\beta+1} \cdot \frac{(f \gamma/\beta)}{(n/p)+((\gamma-1) f/\beta)} 
\end{aligned}
\]
\[
\begin{aligned}
P(\bot)
 &= \frac{1}{\beta+1} \cdot \frac{f \gamma}{(n/p)+((\gamma-1) f/\beta)} 
\end{aligned}
\]
\end{proof}


 \item  \textbf{\AReg Biased, Scheduler Biased}
Under \AReg biased in favour of Byzantine writes and a scheduler biased towards
Byzantine phases by a factor $\gamma \ge 1$, the probability that a correct reader
returns $\bot$ is
\[
    P(\bot) 
        = \frac{(f \gamma/\beta)}{(n/p)+((\gamma-1) f/\beta)} .
\]

\begin{proof}
Since $P(\bot \mid \text{Byzantine phase}) = 1$ and
$P(\text{Byzantine phase}) = \frac{(f \gamma/\beta)}{(n/p)+((\gamma-1) f/\beta)} $,
\[
    P(\bot)
        = 1 \cdot \frac{(f \gamma/\beta)}{(n/p)+((\gamma-1) f/\beta)} 
        = \frac{(f \gamma/\beta)}{(n/p)+((\gamma-1) f/\beta)} .
\]
\end{proof}
\end{enumerate}

\subsection{Resilience Analysis of \AReg}
\label{sec:resilCAReg}
We now derive resilience bounds for the proposed register using the
probabilistic analysis developed in the previous section.


Let $q$ denote the probability that a correct reader returns $\bot$, i.e.,
$P(\bot) = q$.

\begin{enumerate}
  \item  \textbf{Unbiased Atomic Register and Unbiased Scheduler.}
We first consider the setting in which both the atomic register arbitration and
the scheduler are unbiased.
In this case, the probability that a correct reader returns $\bot$ in a phase is $P(\bot) = \frac{fp}{n(\beta+1)}$. Imposing the constraint $P(\bot) =  q$ and solving for $f$ yields
\[
f \le \frac{(\beta+1)qn}{p}.
\]
The register tolerates a number
of Byzantine readers that scales linearly with the system size $n$ and inversely
with the phase size $p$.

\item \textbf{Biased Atomic Register and Biased Scheduler.}
We next consider the worst-case setting in which both the atomic register and the
scheduler are biased towards Byzantine behavior.
In this case, the probability that a correct reader returns $\bot$ in a phase is
\[
P(\bot) = \frac{(f\gamma/\beta)}{(n/p) + ((\gamma - 1)f/\beta)},
\]
where $\gamma \ge 1$ denotes the scheduler bias factor.
Requiring $P(\bot) = q$ and solving for $f$ yields
\[
f \le \frac{qn\beta}{p\bigl(\gamma - q(\gamma - 1)\bigr)}.
\]
This bound is well-defined for all $\gamma \ge 1$ and $q \in [0,1]$, and shows that
scheduler bias directly degrades the resilience of the register.
\end{enumerate}
\subsubsection{\textbf{Resilience Analysis ($\beta = 1$)}}
We consider a setting in which at most one Byzantine reader participates in each Byzantine
phase, i.e., $\beta = 1$.
The resilience of the register is parameterized by $q$ and can be configured according to
application requirements.
For concreteness, we fix $q = \frac{1}{5}$ as an illustrative example in our analysis.

\begin{enumerate}
    \item \textbf{Unbiased Atomic Register and Unbiased Scheduler.}
From the unbiased case,
\[
f = \frac{(\beta+1)qn}{p}.
\]

\item Substituting $\beta = 1$ and $q = \tfrac{1}{5}$, we obtain
\[
f = \frac{2n}{5p}.
\]

\item \textbf{Biased Atomic Register and Biased Scheduler.}
From the biased--biased case,
\[
f = \frac{qn\beta}{p\bigl(\gamma - q(\gamma - 1)\bigr)}.
\]
Setting $\beta = 1$ and $q = \tfrac{1}{5}$ gives

\[
f = \frac{n}{p(4\gamma + 1)}.
\]

\end{enumerate}
\paragraph{Discussion.}
The above analysis demonstrates that phase-based execution significantly limits
Byzantine interference under unbiased \AReg and scheduling, yielding
strong linear resilience guarantees.
\subsubsection{\textbf{Why Lock-Step Execution Yields a Tighter Bound on $P(\bot)$}}

Consider a Byzantine phase containing one correct writer $w$, $\beta$ Byzantine
processes $B_1,\ldots,B_\beta$, and one correct reader $r$. Each of $w$ and
$B_1,\ldots,B_\beta$ may issue a write operation, and $r$ issues a read operation.
Thus, the execution consists of the operations
\[
\{ w, B_1,\ldots,B_\beta, r \}.
\]
All operations are linearized by an unbiased atomic register (\AReg), and all
admissible interleavings are assumed to be equally likely.

\subsubsection*{Case 1: No Lock-Step with the Writer}

In this case, the Byzantine readers do not execute in lock-step with the writer $w$. While the reader $r$ executes in lock-step within the phase, the Byzantine readers progress independently of $w$. Consequently, the operations of $w$, the Byzantine readers, and $r$ may be arbitrarily interleaved in the linearization order. For example, admissible executions include
\[
\begin{aligned}
& W\,B_1\,B_2\,R,\quad
W\,B_1\,R\,B_2,\quad
W\,B_2\,B_1\,R,\\
&W\,B_2\,R\,B_1,\quad
W\,R\,B_1\,B_2,\quad
 W\,R\,B_2\,B_1,\\
&B_1\,W\,B_2\,R,\quad
B_1\,W\,R\,B_2,\quad
B_1\,B_2\,W\,R,\\
&B_1\,B_2\,R\,W,\quad
B_1\,R\,W\,B_2,\quad
B_1\,R\,B_2\,W,\ \ldots
\end{aligned}
\]

\[
(\beta+2)!
\]

A read returns a valid value if and only if no Byzantine write occurs after $w$
and before $r$.

This condition holds in exactly the following two disjoint cases.
\begin{enumerate}
    \item \textbf{Case (Read occurs first).}
The read $r$ precedes all write operations. In this case, the read observes the
most recent valid value written by $w$. The remaining $\beta+1$ write operations
$\{ w, B_1,\ldots,B_\beta \}$ may appear in any order, yielding
\[
(\beta+1)!
\]
executions.
\item \textbf{Case (Writer immediately precedes the read).}
The writer’s write $w$ immediately precedes the read $r$. No Byzantine write can
intervene between them. The $\beta$ Byzantine writes may appear before or after the
contiguous block $wr$ in any order. Treating $wr$ as a single block yields
\[
(\beta+1)!
\]
executions.
\end{enumerate}

Since the two cases are disjoint, the total number of executions in which the
read returns a valid value is
\[
2 (\beta+1)!
\]

Therefore,
\[
P(non \bot \mid \text{no lock-step})
=
\frac{2 (\beta+1)!}{(\beta+2)!}
=
\frac{2}{\beta+2},
\]
and hence
\[
P(\bot \mid \text{no lock-step})
=
1 - \frac{2}{\beta+2}
=
\frac{\beta}{\beta+2}.
\]

\subsubsection*{Case 2: Lock-Step Execution with the Writer}

In the GPU-inspired execution model, all processes execute in lock-step within a
phase. Consequently, the read $r$ occurs after all or before write operations in the
phase. 
\[
\begin{aligned}
r\,(w, B_1, \ldots, B_\beta)
\;\text{or}\;
(w, B_1, \ldots, B_\beta)\,r
\end{aligned}
\]

\[
P(\bot \mid \text{Byzantine phase, \AReg unbiased}) = \frac{\beta}{\beta+1}.
\]
\paragraph{\textbf{Conclusion.}}

Since
\[
\frac{\beta}{\beta+2} < \frac{\beta}{\beta+1},
\]
It follows that lock-step execution strictly increases the probability
that a correct read returns $\bot$ in a Byzantine phase.

In particular, when Byzantine processes execute independently of the writer,
additional interleavings become possible in which the read may occur
outside the interference window.
These executions strictly decrease the adversary’s success probability.

Under lock-step execution, Byzantine processes are forced to compete
directly with the writer within the same phase.
This restriction eliminates favourable interleavings for the reader
and maximises the adversarial impact.

Thus, in this setting, it is adversarially optimal for Byzantine processes
to compete with the writer within the same phase rather than execute independently.

%% file: HWQFULL.tex
\label{subsec:hwqueue-snapshotalgo}

Algorithm~\ref{alg:HWQueue} presents an array-based ordered object inspired by the
Herlihy--Wing concurrent queue.
The object maintains a logical ordering of inserted values using an array of
atomic registers and a single atomic counter \textsc{tail}.
The counter imposes a total order on insert operations, while the array stores
the corresponding values.
$\langle v,\mathit{seq},\mathit{sig}\rangle$ supplied by
Algorithm~\ref{alg:algo2}. As with the other realizations, the object treats
$\mathit{seq}$ and $\mathit{sig}$ as opaque data and performs no authentication
itself.
The \textsc{Insert}$(\langle v,\mathit{seq},\mathit{sig}\rangle)$ operation
atomically increments the counter
\textsc{tail} using a fetch-and-increment primitive to obtain a unique index.
The entry $\langle v,\mathit{seq},\mathit{sig}\rangle$ is then written to the
array position corresponding to this index.
The atomic increment of \textsc{tail} constitutes the linearization point of the
\textsc{Insert} operation and ensures that concurrent inserts are totally
ordered.

The \textsc{ReadLatest} operation first reads the current value of
\textsc{tail}, which determines the logical end of the sequence.
It then performs a bounded scan of the array for a fixed number of steps
proportional to $\Delta+p+1$, collecting values that were written by recent insert
operations. 

The Herlihy--Wing concurrent queue does not have a single linearization point. Since operations span multiple phases, a process may increment the counter but not complete the corresponding insert. Therefore, although the counter advances, some array positions may remain empty. To ensure correctness, readers must obtain the last $p+1$ inserted values; consequently, they must also parse up to $\Delta$ additional positions that may be empty.

The value of $\Delta$ is calculated in the same way as in the memory reclamation argument for the linked-list–based object, assuming that each method call of the Herlihy--Wing concurrent queue completes in at most $\delta$ phases under a strict round-robin scheduler. This bounded traversal guarantees termination while allowing readers to observe a snapshot of recently inserted values.
Since entries in the array are never removed or overwritten, multiple read
operations may concurrently observe the same values without coordination.

The \textsc{ReadLatest} operation can also be implemented without assuming an unbiased round-robin scheduler or that method calls complete within $\delta$ phases. In this case, the \textsc{ReadLatest} method parses the array until either $p+1$ non-empty values are collected or the beginning of the queue is reached.






\begin{algorithm}[h]
\caption{\hwq}
\label{alg:HWQueue}
\begin{algorithmic}[1]

\Statex
\State \textbf{Shared variables:}
\State $\textit{items}[0 \ldots \infty]$ : unbounded array of atomic registers, initially $\bot$
\State $\textit{tail}$ : atomic integer, initially $0$

\Statex
\Procedure{Insert}{$\langle v,\mathit{seq},\mathit{sig}\rangle$}
    \State $i \gets \textit{tail}.\textsc{FetchAndIncrement}()$
    \State $\textit{items}[i] \gets \langle v,\mathit{seq},\mathit{sig}\rangle$
\EndProcedure

\Statex
\Procedure{\readP}{}
    \State $r \gets \textit{tail}.\textsc{Read}()$
    \State $count \gets \Delta + (p+1)$
    \State $i \gets r-1$
    \State $S \gets [\,]$

    \While{$i \geq 0 \land count > 0$}
        \State $x \gets \textit{items}[i]$
        \If{$x \neq \bot$}
            \State $S.\textsc{Append}(x)$ \Comment{$x=\langle v,\mathit{seq},\mathit{sig}\rangle$}
        \EndIf
        \State $i \gets i-1$
        \State $count \gets count-1$
    \EndWhile

    \State \Return $S$ \Comment{entries newest first; reader verifies and selects}
\EndProcedure

\end{algorithmic}
\end{algorithm}
This construction does not perform any memory reclamation.
Once written, array entries remain allocated and are never reused.
As a result, the space usage of the object grows with the number of insert
operations.
This design choice simplifies the implementation and avoids the need for
garbage-collection mechanisms, at the cost of unbounded memory consumption.

The HWQueue-based object relies only on atomic registers and atomic
fetch-and-increment operations.
It provides an ordered abstraction suitable for phase-based execution, where
readers require access only to a bounded number of the most recent values.
\paragraph{\textbf{Correctness of the \hwq}}
\label{subsubsec:correctnessofHW}
We now consider the implementation based on the
Herlihy--Wing concurrent queue (Algorithm~\ref{alg:HWQueue}).

\noindent\textbf{Object-Level Linearizability.}

Algorithm~\ref{alg:HWQueue} is based on the concurrent queue of
Herlihy and Wing~\cite{herlihy1990linearizability}.
It is known that this queue implementation is linearizable
with respect to the FIFO queue specification.

Enqueue operations in this algorithm do not have a fixed
syntactic linearization point.
For each completed \textsc{Insert} operation,
there exists a unique linearization point lying in the
interval between the \textsc{FetchAndIncrement} on
\texttt{tail} and the write to
\texttt{items[$i$]},
as shown in~\cite{herlihy1990linearizability}.

Since Algorithm~\ref{alg:HWQueue} uses only atomic registers and atomic
\textsc{FetchAndIncrement} operations,
it is linearizable under crash failures.
By Lemma~\ref{lem:invocation-independence},
linearizability is preserved even when method
invocations are issued by Byzantine processes.
\\
\paragraph{\textbf{Byzantine Ordering Properties.}}

The ordering arguments for Algorithm~\ref{alg:HWQueue} are identical
to those for the linked-list construction.
The \textsc{FetchAndIncrement} operation on
\texttt{tail} induces a total order on insert operations,
and values become visible to readers only after the
corresponding write to \texttt{items[$i$]}.
Thus, the relative order of observable values
is consistent with write linearization.

Only writer-signed pairs
$\langle v,\mathit{seq}\rangle$ pass verification, and a Byzantine process may
re-insert such a pair but cannot alter its sequence number; hence, as in
Lemma~\ref{lem:ll-latest}, every correct read returns the value of the writer's
most recent write preceding its linearization point. The proofs of
\emph{No Future Read}, \emph{No Read from Distant-Past},
\emph{No New Old Inversion}, and \emph{Intra-Phase Reader Consistency}
therefore carry over directly, against the general adversary.
We omit the details.

%% file: ArrayBasedAtomicObject.tex
\label{subsubsec:array-atomic-object}
Algorithm~\ref{alg:ArrayObject} presents a boun-ded, array-based implementation of the ordered object
\snapqueue. The object maintains a shared array
$A[0 \ldots p]$ and a single atomic counter
\textsc{counter}. The counter assigns logical
timestamps to insert operations, while the array
stores the corresponding signed entries
$\langle v,\mathit{seq},\mathit{sig}\rangle$ supplied by
Algorithm~\ref{alg:algo2}. As with the other realizations, the object treats
$\mathit{seq}$ and $\mathit{sig}$ as opaque data and performs no
authentication itself.

The \textsc{Insert}$(\langle v,\mathit{seq},\mathit{sig}\rangle)$ operation atomically
increments \textsc{counter} using a
fetch-and-increment primitive to obtain a unique
timestamp $ts$. The target position is computed
as $ts \bmod (p+1)$, and the tuple
$\langle v,\mathit{seq},\mathit{sig},ts \rangle$ is written to the
corresponding array slot. The atomic increment
serves as the linearization point and establishes
a total order among insert operations. Since the
array has fixed size $p+1$, insertions overwrite
older entries in a cyclic manner once more than
$p+1$ values have been inserted.

The \textsc{ReadP} operation scans all $p+1$
array entries, collects their contents, and
returns the entries sorted in decreasing order of
timestamp; the reader (Algorithm~\ref{alg:algo2}) then verifies each entry over
$\langle v,\mathit{seq}\rangle$ and returns the valid value with the largest
sequence number.  The timestamps ensure that the
relative order of returned values reflects
the logical insertion order, despite possible
overwriting.

This implementation requires storage for exactly
$p+1$ entries and a single atomic counter.
Correctness relies on the phase-based execution
model: all insert operations of a phase are complete before reads of the subsequent phase begin.
Under this assumption, the array contains
precisely the most recent values relevant to
the current phase.
\begin{algorithm}[h]
\caption{Array-Based Atomic Object}
\label{alg:ArrayObject}
\begin{algorithmic}[1]
\Statex
\State \textbf{Shared variables:}
\State \quad atomic$<$Entry$>$ $\textit{A}[0 \ldots p]$ \Comment{Entry $=\langle \mathit{val},\mathit{seq},\mathit{sig},\mathit{ts}\rangle$}
\State \quad  atomic$<$int$>$ $counter \gets 1$
\Statex
\Procedure{Insert}{$\langle v,\mathit{seq},\mathit{sig}\rangle$}
    \State $ts \gets \textit{counter}.\textsc{Fetch-And-Increment}()$
    \State $pos \gets ts \bmod (p+1)$
    \State $\textit{A}[pos] \gets \langle v,\mathit{seq},\mathit{sig},ts \rangle$
\EndProcedure
\Statex
\Procedure{\readP}{$ $}
    \State $S \gets [\ ]$
    \For{$k = 0$ \textbf{to} $p$}
        \State $S.\textsc{append}(\textit{A}[k])$
    \EndFor
        \State $S \gets \textsc{Sort}(S,ts)$ \Comment{sort $S$ by $ts$ from high to low}
    \State \Return $S$ \Comment{entries; reader verifies and selects}
\EndProcedure
\end{algorithmic}
\end{algorithm}

\paragraph{\textbf{Correctness of the Array-Based Atomic Object}}
\label{subsubsec:correctnessofarrayobject}
We first observe that Algorithm~\ref{alg:ArrayObject} is not linearizable
in the standard asynchronous model.
We then prove that it satisfies Byzantine linearizability
under the phase-based lock-step execution model
illustrated in Figure~\ref{fig:phase-lockstep}.

\medskip
\noindent{\textbf{Non-Linearizability in the Asynchronous Model.}}

Each write operation performs an atomic
\textsc{Fetch-And-Increment} on \texttt{counter}
(Line~5), followed by an atomic write to
$A[\textit{pos}]$ .
Thus writes are individually linearizable.

However, a read operation iterates over the array
sequentially.
If writes are allowed to execute concurrently,
a read may observe some entries before a write
and other entries after the same write.

Since the array is circular
($\textit{pos} = i \bmod (p+1)$),
later writes may overwrite earlier positions.
Consequently, a read may return a set of values
that never coexisted in memory at any single time.

Therefore, no unique linearization point can be assigned
to such a read, and Algorithm~\ref{alg:ArrayObject} is not linearizable
in the asynchronous model.

\medskip
\noindent{\textbf{Phase-Based Lock-Step Execution Model.}}

Execution proceeds in totally ordered phases
(\secref{model}).
Within each phase:

\begin{itemize}
\item All active threads execute in SIMT lock-step.
\item Either a write operation executes,
      or a set of read operations executes.
\item No write overlaps with reads in the same phase.
\end{itemize}

Thus, writes occur in dedicated phases,
and reads execute together in separate phases,
as illustrated in Figure~\ref{fig:phase-lockstep}.
\begin{lemma}
\label{lem:array-phase-linearizable}
Under the phase-based lock-step execution model,
Algorithm~\ref{alg:ArrayObject} is linearizable.
\end{lemma}

\begin{proof}
Phases are totally ordered.

\textbf{Write.}
A write operation executes within a single phase.
Its linearization point is the atomic write to
$A[\textit{pos}]$.

\textbf{Read.}
A read operation executes the loop in
Lines~11--13 in lock-step.
All array reads occur at fixed instruction
positions within the phase. If the write instructions precede the read instructions in the phase, then the read
observes the updated state. If the read instructions precede the write
instructions in the phase, then the read
observes the prior state. In either case, there exists a unique instant
within the phase either just before or just after
the write instruction at which the entire
read result corresponds to the array state. We define that instant as the linearization point
of the read. Thus, every operation admits a unique linearization
point within its phase.
\end{proof}

\begin{lemma}[Latest Value]
\label{lem:array-latest}
Let the writer be correct. Under the phase-based lock-step execution model,
every read executed by a correct process returns the value of the writer's most
recent write preceding its read phase; once the first write has completed, no
correct read returns $\bot$.
\end{lemma}
\begin{proof}
Within a write phase at most $p+1$ insertions occur, one by the writer and at
most $p$ by reader processes (Section~\ref{subsec:design}); their timestamps
are consecutive, so they occupy distinct slots, and the writer's entry of that
phase is present in the array when the subsequent read phase begins. Only pairs
$\langle v,\mathit{seq}\rangle$ signed by the writer pass verification, and a
Byzantine process may re-insert such a pair but cannot alter its sequence
number. As the writer's sequence numbers strictly increase, no valid entry
carries a number exceeding that of its most recent write, while that write's
entry is present and verifies. Hence the selection of the valid entry with the
largest sequence number returns the writer's most recent value, and after the
first write the array always contains a valid entry, so $\bot$ is never
returned.
\end{proof}

\begin{lemma}[No Future Read]
For every correct thread $t$ and phases $i \le j$,
\[
\Readij{i}{j}{t}{}{v} \neq \bot
\;\Rightarrow\;
\exists \writeij{p}{q}{}{v}
\text{ such that } p \le j .
\]
\end{lemma}

\begin{proof}
Since the read is executed by a correct thread and returns $v \neq \bot$,
Lemma~\ref{lem:array-latest} implies that $v$ is the value of the writer's
most recent write preceding the read's phase; denote this write by
$\writeij{p}{q}{}{v}$. Under the phase-based execution model of this
subsection, a write completes within its phase, and the writes are visible to a
read occurs in phases that precede the read phase, so $q < j$. Since every
operation starts no later than the phase in which it completes, $p \le q$,
and therefore $p \le j$.
\end{proof}

\begin{lemma}[No Read from Distant-Past]
Let $\writeij{i}{j}{}{x}$ and $\writeij{k}{l}{}{y}$ be write operations
issued by the correct writer such that $j < k$.
For any correct read $\Readij{p}{q}{t}{}{z}$ with $p > l$,
\[
z \neq x.
\]
\end{lemma}

\begin{proof}
Since $l < p$, the write $\writeij{k}{l}{}{y}$ completes before the read's
phase. By Lemma~\ref{lem:array-latest}, the read returns the value of the
writer's most recent write preceding its read phase, which is therefore
$\writeij{k}{l}{}{y}$ or a later write of the writer; as the writer's
sequence numbers strictly increase, the sequence number of the returned
entry is at least that of $y$, which exceeds that of $x$. A Byzantine
process may re-insert the older pair of $x$, but that entry retains its
original sequence number and is never selected by the reader
(Algorithm~\ref{alg:ArrayObject}). Hence $z \neq x$.
\end{proof}
\begin{lemma}[No New-Old Inversion]
Let
\[
\Readij{i}{j}{t}{a}{x}
\quad\text{and}\quad
\Readij{k}{l}{t'}{b}{x}
\]
be two reads with $j < k$.
Then $b \ge a$.
\end{lemma}

\begin{proof}
Since phases are totally ordered,
\[
\ell(\Readthreeij{i}{j}{t}) <
\ell(\Readthreeij{k}{l}{t'}).
\]

Case 1: No write occurs between the two reads.
Both reads observe the same array state.
Hence $b = a$.

Case 2: At least one write occurs between them.
The later read observes the updated array state.
Since writes are ordered by the
\textsc{Fetch-And-Increment} counter,
the later read cannot observe an older value
than the earlier read.
Hence $b \ge a$.
\end{proof}
\begin{lemma}[Intra-Phase Reader Consistency]
All correct reads that complete in the same phase
return the same value.
\end{lemma}

\begin{proof}
Within a read phase, all threads execute the
loop in Lines~11--13 in SIMT lock-step.
No write occurs in that phase.
Hence all reads observe the identical array state
and produce identical outputs.
\end{proof}
\begin{theorem}
Under the phase-based lock-step execution model,
Algorithm~\ref{alg:ArrayObject} implements a Byzantine-linearizable
SWMR register.
\end{theorem}

\begin{proof}
Lemma~\ref{lem:array-phase-linearizable}
establishes linearizability.
The four ordering lemmas establish the
conditions required by Definition~\ref{def:byzlinz}.
Hence the construction satisfies Byzantine
linearizability under the phase model.
\end{proof}

\medskip
\noindent{\textbf{Correctness of \snapqueue}}
\label{subsec:correctnessofsnapqueueobject}

We establish that the \snapqueue abstraction is
Byzantine linearizable in the GPU-inspired model.
Specifically, we prove that each implementation of
\snapqueue is linearizable with respect to the
sequential specification, and that its ordered semantics together with
phase-based lock-step execution ensure the properties
required for Byzantine linearizability.

We consider four implementations: the linked-list-based
(Algorithm~\ref{alg:LinkedList}), the wait-free circular array
(Algorithm~\ref{alg:wfarray}), the \hwq\ (Algorithm~\ref{alg:HWQueue}),
and the array-based construction (Algorithm~\ref{alg:ArrayObject}).
The first three are linearizable under concurrent execution, while the
array-based construction is linearizable if its operations do not span
multiple phases.

Since each implementation is linearizable and preserves
the required ordering guarantees under the execution
model, it follows that the SnapObject abstraction is
Byzantine linearizable. Consequently, any of these
implementations can serve as a building block for a
Byzantine-linearizable SWMR register.

%% file: WaitFreeSnapObject.tex
\subsubsection{\textbf{Memory Reclamation for the Linked-List Object}}
\label{app:reclamation}

\noindent The linked-list–based SnapObject \algoref{LinkedList} generates memory which can be reclaimed. Algorithm~\ref{alg:memoryReclaim} presents an optional garbage-collection procedure for this object. Starting from a given node, the procedure traverses the list up to the bound $\Delta = (p+1)(\delta-1)\frac{n}{p}$, where $\delta$ denotes the maximum number of scheduled phases required for a process to complete an operation on the linked-list object. If the list exceeds this bound, the $\Delta$-th node’s \texttt{next} pointer is atomically set to \texttt{NULL} using CAS (line~\ref{lin:temp}), thereby detaching the suffix, which is then reclaimed via \textsc{FreeList} \mth (which recycles the nodes). The CAS guarantees safety under concurrent insertions. Figure~\ref{fig:reclamation} illustrates this truncation. The choice of $\Delta$ is derived below. 

\begin{algorithm}[h]
\small
\caption{Memory Reclamation for Linked-List–Based Atomic Object}
\label{alg:memoryReclaim}
\begin{algorithmic}[1]
\Statex
\State \textbf{Node structure:} $\langle  \textit{val}, \textit{seq}, \textit{sig}, \textit{next}\rangle$

\State \textbf{Constant:} $\Delta = (p+1)(\delta-1)\frac{n}{p}$

\Procedure{GarbageCollection}{$curr$}
        \State $count \gets 1$
        \While{$curr \neq NULL \land count < \Delta$}
            \State $prev \gets curr$
            \State $curr \gets curr.\textit{next}$
            \State $count \gets count + 1$
        \EndWhile
        \If{$curr \neq NULL$}
            \State $oldNext \gets prev.\textit{next}$
            \label{lin:temp}
            \State $prev.\textit{next}.\textsc{CAS}(oldNext, NULL)$
            \State \Call{FreeList}{$oldNext$}
        \EndIf
\EndProcedure

\end{algorithmic}
\end{algorithm}

\vspace{-0.5em}
\noindent\textbf{Computation of the Cleanup Bound $\mathbf{\Delta}$.}
Let us assume that the scheduler guarantees a process that starts its execution in phase $i$ and requires
$\delta$ scheduled phases to complete the linked list methods.
Under unbiased round-robin scheduling, since $p$ reader processes are scheduled per phase, each individual process is scheduled once every $\frac{n}{p}$ phases.
Therefore, a process that starts in phase $i$ executes in the  phases $
i,\;
i + \frac{n}{p},\;
i + 2\frac{n}{p},\;
\ldots,\;
i + (\delta-1)\frac{n}{p}.
$

Hence, it completes in phase $i + (\delta-1)\frac{n}{p}$.

\paragraph{Processes Active in Phase $j$.}

A process that started in phase $s$ is still active in phase $j$ if $
s + (\delta-1)\frac{n}{p} \ge j$.
Solving for $s$ yields $
s \ge j - (\delta-1)\frac{n}{p}$.
Thus, any process that started at or after phase $
j - (\delta-1)\frac{n}{p}$
 may still be active in phase $j$.

\paragraph{Number of Active Start Phases.}

The number of such start phases is $j - \left(j - (\delta-1)\frac{n}{p}\right)
=
(\delta-1)\frac{n}{p}$.
\begin{figure} [h]
    \centering
    \includegraphics[scale=0.41]{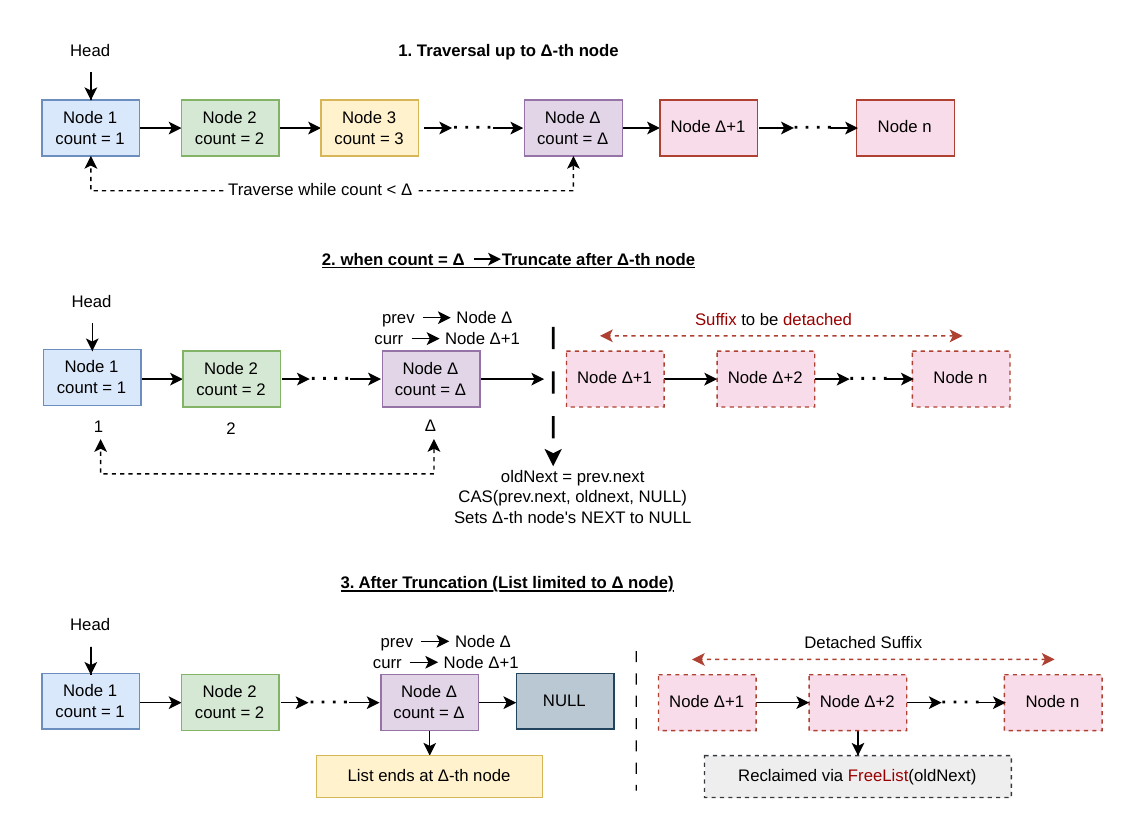}
    \vspace{-1.5em}
    \caption{Memory reclamation for the linked-list–based object with bound $(p+1)(\delta-1)\frac{n}{p}$. The list is traversed until the $\Delta$-th node, after which a CAS operation sets the $\Delta$-th node’s \texttt{next} pointer to \texttt{NULL}, detaching and reclaiming the suffix via \texttt{FreeList(oldNext)}.}

    \label{fig:reclamation}
\end{figure}
\paragraph{Worst-Case Retained Nodes.}

Since at most $p$ reader processes execute per phase, at most $p+1$ nodes may
be accessed (or written) in each phase. The writer would only insert but not read so it will not be accessing any nodes except head. Therefore, the total number of
nodes that may still be accessed by processes active in phase $j$ is $
\Delta
=
(p+1) \cdot (\delta-1)\frac{n}{p}$.

Therefore, the memory complexity when $n$ is much larger than both $p$ and $\delta$ is $O(n)
$.

Retaining $(p+1)(\delta-1)n/p$ nodes suffices to guarantee that no node that may
still be accessed by an active process is reclaimed prematurely. If we define a phase as the smallest unit of progress, namely a single memory operation (one fetch or one store instruction). Under this definition, and as shown in \secref{upperBound}, for $p = 32$ we obtain $\delta = 201$, which yields $\Delta = 207n$.

\subsubsection{\textbf{Wait-Free Circular-Array Based Atomic Object}}
\label{subsubsec:wfarray-object}
 
The linked-list realization (Algorithm~\ref{alg:LinkedList}) is lock-free but
not wait-free: its \textsc{Insert} is a CAS-retry loop, and under a register
biased toward Byzantine writes a Byzantine process may win every round of
contention, so the writer may retry without bound. The array-based realization
(Algorithm~\ref{alg:ArrayObject}) has no retry loop but is correct only when
operations complete within a single phase. We now give a realization, WF-Array,
whose \textsc{Insert} consists of two unconditional shared-memory operations
and which remains correct when operations span multiple phases.
 
\algoref{wfarray} presents the construction. The object maintains an atomic
counter $\mathit{tail}$ (a ticket dispenser) and a shared circular array
$A[0\ldots K-1]$ of atomic registers, each holding an entry
$\langle \mathit{val},\mathit{seq},\mathit{sig},\mathit{tkt}\rangle$, where
$\langle \mathit{val},\mathit{seq},\mathit{sig}\rangle$ is the signed entry
supplied by Algorithm~\ref{alg:algo2} and $\mathit{tkt}$ the ticket assigned by
the object. As with the linked-list realization, the object treats
$\mathit{seq}$ and $\mathit{sig}$ as opaque data and performs no authentication
itself; the ticket serves only to place the entry in the ring and to reject
stale entries from an earlier lap.
 
\medskip\noindent\textit{\em Insert Operation.}
$\textsc{Insert}(\langle v,\mathit{seq},\mathit{sig}\rangle)$ (i) performs a
fetch-and-increment on $\mathit{tail}$ to obtain a unique ticket $t$, and (ii)
writes $\langle v,\mathit{seq},\mathit{sig},t\rangle$ to position $t \bmod K$.
Both steps are unconditional and cannot fail; the store is the linearization
point.

\medskip\noindent\textit{\em ReadLatest Operation.}
$\textsc{\readP}()$ reads $\mathit{tail}$ once to obtain the current logical end
$T$ and scans the $W$ most recently issued tickets, from newest ($T-1$) to
oldest ($T-W$). It keeps an entry only if the slot is non-empty and its ticket
equals the scanned position ($\mathit{tkt}=T-a$), which rejects stale entries
from an earlier lap, and returns the collected entries, newest first. The
reader (Algorithm~\ref{alg:algo2}) then verifies each entry over
$\langle \mathit{val},\mathit{seq}\rangle$ and returns the valid value with the
largest sequence number; if none verifies, the read returns $\bot$.

\paragraph{\textbf{Assumptions.}}
Besides the assumptions of Section~\ref{subsec:CARegAssumption}, the analysis
uses the following.
\begin{enumerate}
  \item[(R1)] \emph{Phase granularity.} Each scheduled phase admits at least
  two shared-memory operations per process, and within a phase the insert group
  executes before the read group (Figure~\ref{fig:phase-lockstep}). Hence an
  insert invoked by a correct process completes in the phase in which it is
  invoked.
  \item[(R2)] \emph{Bounded insert rate.} A process completes at most $x$
  inserts in a phase in which it is scheduled ($x=1$ under
  Section~\ref{subsec:design}). Taking the worst case in which all $p$ readers
  scheduled in a phase are Byzantine, at most $px+1$ tickets are issued per
  phase: $px$ by readers and one by the writer.
  \item[(R3)] \emph{Fair scheduling.} The scheduler is round-robin: each reader
  is scheduled once every $\lceil n/p\rceil$ phases, and the writer participates
  in every phase.
  \item[(R4)] \emph{$\delta$-bounded stores.} The store of an insert
  (line~\ref{lin:wfstore}) is performed within $\delta \le W+1$ scheduled phases
  of its fetch-and-increment, or not at all; a store that never lands is skipped
  by readers, exactly as a crash during the insert.
\end{enumerate}
 
\noindent The window and the array size are set to
\[
  W = px+1,
  \qquad
  K = (W+1)\big\lceil n/p\big\rceil^{2}.
\]
 
\paragraph{\textbf{Wait-Freedom and Linearization.}}
 
\begin{lemma}[Wait-Freedom]
\label{lem:wfarray-wf}
Every \textsc{Insert} completes in exactly two shared-memory operations, and
every \textsc{\readP} in exactly $W+1$ shared-memory operations, regardless of
the scheduler and of the behavior of the other processes.
\end{lemma}
\begin{proof}
\textsc{Insert} performs one fetch-and-increment (line~\ref{lin:wffai}) and one
store (line~\ref{lin:wfstore}); the fetch-and-increment is unconditional and
concurrent invocations are ordered by the register's arbitration, and the store
targets the single slot fixed by the ticket. \textsc{\readP} performs one read
of $\mathit{tail}$ (line~\ref{lin:wftail}) followed by exactly $W$ loads
(line~\ref{lin:wfload}). Neither procedure contains a retry loop, so the bounds
hold under any scheduler and any Byzantine behavior.
\end{proof}
 
\begin{lemma}[Linearization Points]
\label{lem:wfarray-lp}
Every completed \textsc{Insert} linearizes at its store
(line~\ref{lin:wfstore}), and every completed \textsc{\readP} at its read of
$\mathit{tail}$ (line~\ref{lin:wftail}). Ordering completed operations by these
instants yields a total order that respects real-time precedence.
\end{lemma}
\begin{proof}
Each of the two events is a single atomic memory operation and therefore occurs
at a unique instant between the invocation and the response of its operation.
Ordering by these instants respects real-time order. By
Lemma~\ref{lem:invocation-independence}, this holds regardless of whether
invocations are issued by correct or Byzantine processes.
\end{proof}
 
\begin{lemma}[Ticket Order]
\label{lem:wfarray-order}
The tickets of the correct writer are strictly increasing, and its stores occur
in ticket order. Moreover, if a \textsc{\readP} reads $\mathit{tail}=T$, then
every writer ticket $\tau \le T-1$ has performed its store before that read of
$\mathit{tail}$.
\end{lemma}
\begin{proof}
The writer is sequential, so its fetch-and-increments, hence its
tickets, strictly increase, and each of its stores precedes the
fetch-and-increment of its next write. Consider a writer ticket $\tau \le T-1$.
Its fetch-and-increment precedes the read of $\mathit{tail}$, since
$\mathit{tail}$ is advanced only by fetch-and-increment and already exceeds
$\tau$ when read. If that fetch-and-increment occurs in an earlier phase, its
store completes in the same phase by (R1) and hence precedes the read. If it
occurs in the read's phase, then by (R1) the insert group precedes the read
group, so the store again precedes the read of $\mathit{tail}$.
\end{proof}
 
\begin{lemma}[Read Correctness]
\label{lem:wfarray-read}
If the writer is correct, then a read by a correct process whose
\textsc{\readP} reads $\mathit{tail}=T$ returns the value of the writer's write
with the largest ticket $\tau^{*}\le T-1$, which is the writer's most recent
write preceding the read's linearization point, and never returns $\bot$.
\end{lemma}
\begin{proof}
Since the writer inserts in every phase before reads execute (R1, R3), its
latest ticket $\tau^{*}$ satisfies $T-W \le \tau^{*} \le T-1$, so its entry
lies in the scanned window and, by the choice of $K$, is not yet overwritten
by a later lap. Byzantine processes cannot forge the writer's signature, and
a replayed writer-signed entry keeps its old sequence number, so no valid
entry in $S$ carries a sequence number exceeding $\mathit{seq}^{*}$. Hence
the largest-$\mathit{seq}$ selection returns $v^{*}$, the writer's most
recent value, and never $\bot$.
\end{proof}
\paragraph{\textbf{Byzantine Ordering Properties.}}
We now establish the four conditions of Definition~\ref{def:byzlinz}.
 
\begin{lemma}[No Future Read]
\label{lem:wfarray-nfr}
For every correct thread $t$ and phases $i \le j$,
$\Readij{i}{j}{t}{}{v}\neq\bot \Rightarrow \exists\,\writeij{p}{q}{}{v}$ with
$p \le j$.
\end{lemma}
\begin{proof}
Since the read is executed by a correct thread and returns $v \neq \bot$,
Lemma~\ref{lem:wfarray-read} implies that $v$ is the value of the writer's most
recent write whose linearization point precedes that of the read; denote this
write by $\writeij{p}{q}{}{v}$. Thus
\[
\ell(\writeij{p}{q}{}{v}) < \ell(\Readij{i}{j}{t}{}{v}).
\]
By Lemma~\ref{lem:lp-localization}, every operation linearizes between its
invocation and its response, so $p \le \ell(\writetwoij{p}{q})$ and
$\ell(\Readthreeij{i}{j}{t}) \le j$. Combining the three inequalities yields
$p \le j$.
\end{proof}
 
\begin{lemma}[No Read from Distant-Past]
\label{lem:wfarray-nrdp}
Let $\writeij{i}{j}{u}{x}$ and $\writeij{k}{l}{v}{x}$ be writes of the correct
writer with $\ell(\writeij{i}{j}{u}{x})<\ell(\writeij{k}{l}{v}{x})$. For any
correct read $\Readthreeij{p}{q}{t}$ with
$\ell(\writeij{k}{l}{v}{x})<\ell(\Readthreeij{p}{q}{t})$, the returned value is
not $x_u$.
\end{lemma}
\begin{proof}
Let $\tau_u,\tau_v$ be the tickets of the two writes; by
Lemma~\ref{lem:wfarray-order}, $\tau_u<\tau_v$. The store of $x_v$ precedes the
read's linearization point, at which $\mathit{tail}=T$; since the
fetch-and-increment issuing $\tau_v$ precedes that store, $\tau_v \le T-1$. By
Lemma~\ref{lem:wfarray-read} the read returns the value of the largest
writer ticket at most $T-1$, which is $\ge \tau_v > \tau_u$. As the writer's
values are distinct, the returned value is not $x_u$.
\end{proof}
 
\begin{lemma}[No New--Old Inversion]
\label{lem:wfarray-nnoi}
Let $\Readij{i}{j}{t}{a}{x}$ and $\Readij{k}{l}{t'}{b}{x}$ be reads by correct
threads with $j<k$. Then $b \ge a$.
\end{lemma}
\begin{proof}
Since $j<k$, the first read completes before the second is invoked, so
$\ell(\Readthreeij{i}{j}{t})<\ell(\Readthreeij{k}{l}{t'})$. Let $T_1,T_2$ be the
values of $\mathit{tail}$ they read; as $\mathit{tail}$ is advanced only by
fetch-and-increment, it is non-decreasing, so $T_1 \le T_2$. By
Lemma~\ref{lem:wfarray-read} the reads return the writer values with the
largest tickets $\tau_a \le T_1-1$ and $\tau_b \le T_2-1$; hence
$\tau_b \ge \tau_a$, and since the writer writes in ticket order, $b \ge a$.
\end{proof}
 
\begin{lemma}[Intra-Phase Reader Consistency]
\label{lem:wfarray-ipc}
Let $j$ be any phase and let $\Readthreeij{i}{j}{1},\ldots,\Readthreeij{i}{j}{p}$
be reads by correct threads that all complete in phase $j$. Then all return the
same value.
\end{lemma}
\begin{proof}
Exactly one warp is scheduled per phase, so all reads completing in phase $j$
belong to the same warp. By SIMT lock-step execution (Section~\ref{sec:model}),
these threads read $\mathit{tail}$ at the same instant and obtain the same $T$,
and at each iteration load the same position $(T-a)\bmod K$ and observe the same
contents; the filter of line~\ref{lin:wfmatch}, the signature verification over
$\langle \mathit{val},\mathit{seq}\rangle$, and the selection of the valid value
with the largest sequence number (Algorithm~\ref{alg:algo2}) are deterministic
functions of these values. Hence all such reads return the same value.
\end{proof}
\paragraph{\textbf{Byzantine Linearizability and Resilience.}}
 
\begin{theorem}
\label{thm:wfarray-byzlin}
\algoref{wfarray} implements a Byzantine-linearizable SWMR register.
\end{theorem}
\begin{proof}
By Lemma~\ref{lem:wfarray-lp} the completed operations admit a total order
consistent with real-time precedence, and by Lemma~\ref{lem:wfarray-read} each
correct read returns the value of the writer's most recent write preceding its
linearization point; hence the induced sequential history is legal, and the
history observed by the correct processes extends to a legal linearizable
history, as required by Definition~\ref{def:iditDef}. By
Lemma~\ref{lem:invocation-independence}, this holds under Byzantine
invocations. Lemmas~\ref{lem:wfarray-nfr}--\ref{lem:wfarray-ipc} establish the
four ordering conditions of Definition~\ref{def:byzlinz} against the general
adversary. Hence, the construction is Byzantine-linearizable.
\end{proof}
 
\begin{theorem}
\label{thm:wfarray-resilience}
If the writer is correct and $C \neq \emptyset$, then the WF-Array construction
is Byzantine-linearizable for all $f \le n-1$, and for every correct reader $r$,
$P(r\text{ returns }\bot)=0$.
\end{theorem}
\begin{proof}
The second claim is Lemma~\ref{lem:wfarray-read}. For the first, no step of
the analysis bounds the number of Byzantine readers: Byzantine interference
enters only through the number of tickets issued per phase, at most $px+1$, and
this is absorbed by the choice of $W$ and $K$, including the worst case in which
all $p$ readers of a phase are Byzantine. Correctness of a read depends only on
the writer being correct (Lemma~\ref{lem:wfarray-read}), so a single correct
reader suffices, giving $f \le n-1$.
\end{proof}

%% file: queueanalysis.tex
\label{subsubsec:resilienceofsnapqueue}
Let $W$ denote the writer and
$R = \{r_1,\dots,r_n\}$ the set of readers.
Let $F \subseteq R$ be the set of Byzantine readers with $|F| = f$,
and let $C = R \setminus F$ denote the set of correct readers.

Assume $W$ is correct and participates in every phase.

\medskip
\noindent
{\textbf{Probability of Returning $\bot$.}}

Let $P(\bot)$ denote the probability that a correct read returns $\bot$.

\begin{theorem}
If $W$ is correct, then for every execution and every $r_i \in C$, $P(r_i \text{ returns } \bot) = 0$.
\end{theorem}

\begin{proof}
In every phase, the writer invokes
$\textsc{Insert}(\langle v,\mathit{seq},\mathit{sig}\rangle)$, where
$\mathit{sig}$ verifies over the pair $\langle v,\mathit{seq}\rangle$ under
$W$'s public key.

At most $p+1$ insert operations occur per phase, given that $p$ reader
processes and the writer are present per phase and each process can call
insert only once in a phase duration.
The operation $\textsc{ReadLatest}$ returns the $p+1$ most recent inserted
entries (a larger window in the \wf\ realizations).
Since the writer inserts a validly signed entry in every phase,
the returned set always contains at least one entry
that verifies under $W$'s public key.

By unforgeability of signatures,
no Byzantine reader can produce an entry
that verifies under $W$'s public key.
Hence every correct reader selects a valid writer-generated value, namely the
one carrying the largest sequence number.

Therefore, $
\forall r \in C :
P(r \text{ returns } \bot) = 0.
$
\end{proof}

\noindent
{\textbf{Resilience Factor.}}

\begin{theorem}
If $W$ is correct and $C \neq \emptyset$,
then the SnapObject-based construction
is Byzantine-linearizable for all $f \le n - 1$.
\end{theorem}

\begin{proof}
Suppose $|C| \ge 1$ and let $r \in C$.

In every phase, the writer inserts a value that verifies under its public key.
The operation $\textsc{ReadLatest}$ returns the $p+1$ most recent
insertions from the underlying SnapObject.
Reader $r$ verifies the signature of each returned value using
$W$'s public key and discards any value whose signature is invalid.
By unforgeability of digital signatures, no Byzantine process can
produce a value that verifies under $W$'s public key.
Since $W$ is correct and continuously inserts properly signed values,
every $p+1$ values returned by $\textsc{ReadLatest}$
contains at least one value produced by $W$.
Thus, among the returned candidates, at least one value passes
signature verification.

Thus correctness depends only on the existence of
at least one correct reader,
which is equivalent to $f \le n - 1$.
\end{proof}

\noindent\emph{Remark (replay).} Unlike the CAReg construction of
Section~\ref{sec:atReg}, whose freshness properties hold only against a
non-replaying adversary (Appendix~\ref{sec:CCAeg}), the register of
Algorithm~\ref{alg:algo2} satisfies the ordering properties of
Definition~\ref{def:byzlinz} against the general adversary, under every
realization of \snapqueue: a Byzantine process may re-insert an older
writer-signed entry into the object, but the entry retains its original
sequence number and is therefore never selected by the reader.

%% file: SignatureAware.tex
\subsubsection{\textbf{Design Features}}
Following assumptions are made for the algorithmic solution:
\begin{enumerate}
    \item Shared Memory
    \item \ARegister (\AReg)
    \item Signatures
\end{enumerate}

\subsubsection{\textbf{Algorithmic Solution}}
\label{subsec:ideal_algo_sol}
\paragraph{Algorithm Overview.}
The signature-aware memory object maintains a single atomic register
\texttt{reg}, initially set to \texttt{NULL}. The object exports two operations:
\texttt{Write} and \texttt{Read}.
\begin{algorithm}[h]
\caption{Signature-Aware Single-Memory Object}
\label{alg:ideal}
\begin{algorithmic}[1]
\Statex
\State \textbf{Shared variable:}
\State \quad $atomic<memory>\ reg \gets \langle NULL, 0 \rangle$ \Comment{$reg = \langle val, seq \rangle$}
\Statex
\Procedure{Write}{$\langle v, seq, sig \rangle$}
    \If{$\textsc{Verify}(\langle v, seq \rangle, sig)\ \wedge\ seq > reg.seq$}
        \State $reg \gets \langle v, seq \rangle$
    \EndIf
\EndProcedure
\Statex
\Procedure{Read}{$ $}
    \State \Return $reg.val$
\EndProcedure
\end{algorithmic}
\end{algorithm}


A \texttt{Write} operation takes as input a triple $\langle v, seq, sig \rangle$, where $v$ is the value to be written, $seq$ is the writer's sequence number, and $sig$ is a digital signature over the pair $\langle v, seq \rangle$, exactly as produced by the writer in Algorithm~2. The memory object maintains, alongside the stored value, the sequence number of the most recently accepted write. Upon invocation, the memory verifies the signature using a deterministic verification procedure \texttt{Verify} and additionally checks that the supplied sequence number strictly exceeds the stored one. If both checks pass, the register is atomically updated to $\langle v, seq \rangle$; otherwise, the write has no effect.

A \texttt{Read} operation simply returns the value currently stored in \texttt{reg}.

The key property of this construction is that both authentication and freshness are enforced internally by the memory object. Signature verification ensures that Byzantine processes cannot install forged values, while the monotonicity check on the stored counter ensures that they cannot re-install old ones: since the sequence number is bound inside the signature, a Byzantine process can neither forge a fresh counter nor strip the counter from a replayed pair, so a replay necessarily carries a stale sequence number and is rejected. Consequently, the register state can be modified only by the writer's genuinely new values, and all guarantees of this construction hold against the general adversary of Section~\ref{sec:model}.

\subsubsection{\textbf{Correctness for Idealized Signature-Aware Memory}}
\label{subsubsec:correctnessofsigawarememory}
We show that Algorithm~\ref{alg:ideal} implements a Byzantine-linearizable
single-writer multi-reader (SWMR) register according to Definition~\ref{def:byzlinz}.

\medskip
\noindent\textbf{Proof Structure.}
The proof follows the same structure as for CAReg.
We first establish object-level linearizability.
We then prove writer authenticity.
Finally, we show that the ordering properties required
for Byzantine linearizability follow from atomic
register semantics.

\medskip
\noindent\textbf{Object-Level Linearizability}

\begin{lemma}[Signature-Aware Register Linearizability]
\label{lem:sigmem-linearizable}
All completed read and write operations on the
signature-aware memory object (Algorithm~\ref{alg:ideal})
are linearizable in the Byzantine GPU-inspired
shared-memory model.
\end{lemma}

\begin{proof}
Algorithm~\ref{alg:ideal} maintains a single atomic register \texttt{reg}.
A successful \textsc{Write} operation atomically updates
\texttt{reg} (Line~5).
A \textsc{Read} operation atomically reads \texttt{reg}
(Line~9).
A write whose signature verification fails
(Line~4) has no effect on the state.

The linearization point of a successful write
is the atomic assignment to \texttt{reg},
and the linearization point of a read
is the atomic read of \texttt{reg}.

Since atomic read/write registers are linearizable
under crash failures, and linearizability is preserved
under Byzantine invocations by Lemma ~\ref{lem:invocation-independence},
Algorithm~\ref{alg:ideal} is linearizable.
\end{proof}

\medskip
\noindent\textbf{Value Validity}

\begin{lemma}[Writer Authenticity]
\label{lem:sigmem-authenticity}
If a correct read returns a value $v \neq \texttt{NULL}$,
then $v$ was written by the authorized writer.
\end{lemma}

\begin{proof}
By Algorithm~\ref{alg:ideal} (Lines~4--5), the register is updated
only if the supplied signature verifies successfully
under the authorized writer's public key.
Since Byzantine processes cannot forge valid signatures,
only values generated by the authorized writer
can modify \texttt{reg}.
\end{proof}

\begin{lemma}[Replay Resistance]
\label{lem:sigmem-freshness}
No write invocation by a Byzantine process modifies \texttt{reg}. In particular, re-installing a previously signed pair has no effect, and \texttt{reg.seq} is strictly increasing over successful writes.
\end{lemma}

\begin{proof}
A write modifies \texttt{reg} only if $\textsc{Verify}(\langle v, seq \rangle, sig)$ succeeds and $seq > \texttt{reg.seq}$ (Line~4). By unforgeability, a Byzantine process cannot produce a valid signature over any pair the writer has not signed; the only verifiable triples available to it are copies of the writer's past writes. Each such triple carries the sequence number the writer bound to it at signing time, which cannot be altered without invalidating the signature. Since the writer's sequence numbers are strictly increasing and every writer-issued write with that number has either already been applied or superseded, a replayed triple satisfies $seq \le \texttt{reg.seq}$ and fails the guard. Hence every successful write originates from a fresh invocation by the writer, and \texttt{reg.seq} strictly increases.
\end{proof}

\subsubsection{Byzantine Linearizability}

We now show that Algorithm~\ref{alg:ideal} satisfies the ordering
properties required by Definition~\ref{def:byzlinz}.

Because the object is a single atomic register,
all writes are totally ordered by their linearization points,
and each read returns the value written by the most recent
preceding write.

Thus, the proofs of
\emph{No Future Read},
\emph{No Read from Distant-Past},
\emph{No New--Old Inversion}, and
\emph{Intra-Phase Reader Consistency}
are identical to those given for CAReg
(Lemmas~\ref{thm:no-future}--\ref{thm:intra-phase}),
replacing references to CAReg with Algorithm~\ref{alg:ideal}
and using Lemma~\ref{lem:sigmem-linearizable}
in place of Lemma~\ref{lem:careg-linearizable}.

We therefore omit the details.
\begin{theorem}
Algorithm~\ref{alg:ideal} implements a Byzantine-linearizable
SWMR register.
\end{theorem}

\begin{proof}
Lemma~\ref{lem:sigmem-linearizable}
establishes object-level linearizability.
Lemma~\ref{lem:sigmem-authenticity}
ensures writer authenticity.
The required ordering properties follow
from atomic register semantics as argued above.
Hence, Algorithm~\ref{alg:ideal} is Byzantine-linearizable.
\end{proof}

\subsubsection{\textbf{Resilience of the Idealized Signature-Aware Memory}}
\label{subsubsec:resilienceofsigawarememory}
Let $W$ denote the writer and
$R = \{r_1,\dots,r_n\}$ the set of readers.
Let $F \subseteq R$ be the set of Byzantine readers.

Assume $W$ is correct.

\medskip
\noindent\textbf{Probability of Returning $\bot$.}

\begin{theorem}
If $W$ is correct, then for every execution and every
$r \in R \setminus F$,
\[
P(r \text{ returns } \bot) = 0.
\]
\end{theorem}

\begin{proof}
The memory abstraction updates its state only if
a write request verifies under $W$'s public key.
Invalid write attempts have no effect.

Hence the register state always equals
the most recent value written by $W$.
Since Byzantine readers cannot alter the state,
every correct read returns a writer-generated value.

Therefore,
\[
\forall r \in R \setminus F :
P(r \text{ returns } \bot) = 0.
\]
\end{proof}

\medskip
\noindent\textbf{Resilience Factor.}

\begin{theorem}
If $W$ is correct and $R \setminus F \neq \emptyset$,
then the construction tolerates
\[
f \le n - 1.
\]
\end{theorem}

\begin{proof}
By construction, the memory object accepts write operations
only if they are properly authorised by $W$.
Since $W$ is correct, all accepted writes originate from $W$.
By the unforgeability of authentication, no Byzantine process can
successfully perform a write that is accepted by the memory object.

Readers, whether correct or Byzantine, have no ability to modify
the stored value. Therefore, the state of the register is determined
exclusively by the sequence of writes performed by $W$.

The correct reader $r$ observes only values that were previously
written by $W$. Hence every value returned by a correct reader
corresponds to a real write operation, and the sequential order
of the register is induced solely by $W$'s writes.

Therefore correctness requires only that
$R \setminus F \neq \emptyset$,
which is equivalent to $f \le n - 1$.
\end{proof}